\documentclass[12pt,a4paper,ngerman]{article}
\usepackage{graphicx}
\usepackage{amssymb}

\usepackage{pstricks-add}
\usepackage{pst-plot}
\usepackage{bbm}
\usepackage{inputenc}
\usepackage{ae}
\usepackage{amsmath}
\usepackage{amsthm}
\usepackage{wasysym}
\usepackage{multirow}
\usepackage{subfigure}
\usepackage[round]{natbib}
\usepackage{framed}
\usepackage{makeidx}
\usepackage{array}
\usepackage{verbatim}
\usepackage{pstricks}
\usepackage{longtable}
\usepackage{color}
\usepackage{colortbl}
\usepackage{hyperref}
\makeindex
\DeclareSymbolFont{calsym}{OMS}{ztmcm}{m}{n}
\DeclareSymbolFontAlphabet{\mathcal}{calsym}
\newcommand{\bb}{\mathbb }

\newcommand{\ird}{\int_{\bb{R}^d}}
\newcommand{\td}{\tilde }
\newcommand{\ol}{\overline }
\newcommand{\ca}{\mathcal }

\newcommand{\tn}{\, \textnormal}

\newcommand{\beq}{\begin{equation}\label }
\newcommand{\eeq}{\end{equation} }
\newcommand{\bal}{\begin{align*} }
\newcommand{\ve}{\varepsilon}

\newcommand{\bs}{\boldsymbol}
\newcommand{\eal}{\end{align*} }

\newtheorem{ass}{Assumption}[]
\newtheorem{satz}{Theorem}[section]
\newtheorem{lemma}[satz]{Lemma}

\newenvironment{myproof}[1][\proofname]{\proof[#1]\mbox{}\\*}{\endproof}
\renewcommand{\proofname}{Proof}
\renewcommand\proof[1][\proofname]{%
\par
\pushQED{\qed}%
\normalfont \topsep 6pt plus 6pt\relax
\trivlist
\item[\hskip\labelsep
      #1:]\ignorespaces}
\renewcommand{\arraystretch}{1.1}
\makeindex
\begin{document}
\setlength\parindent{0pt}
\parindent 0cm

\title{Multiscale inference for multivariate deconvolution}

\author{
{\small Konstantin Eckle, Nicolai Bissantz, Holger Dette } \\
{\small Ruhr-Universit\"at Bochum } \\
{\small Fakult\"at f\"ur Mathematik } \\
{\small 44780 Bochum, Germany } 
}
\date{}

\maketitle
\begin{abstract}

In this paper we provide new methodology for  inference of  
the geometric features of a multivariate density
in deconvolution.
Our approach is based on multiscale tests to detect signi\-ficant  directional derivatives of the 
unknown density at arbitrary points in arbitrary directions. 
The multiscale method is used to  identify regions of monotonicity and to construct a general
procedure for the  detection  of modes of the multivariate density. Moreover,
as an important application a significance test for the presence 
of a local maximum at a pre-specified  point is proposed. 
 The performance of the new  methods is investigated from a 
theoretical point of view and  the finite sample properties are illustrated by means of a small simulation study.

\end{abstract}



\section{Introduction}
\def\theequation{1.\arabic{equation}}
\setcounter{equation}{0}
In many  applications such as  in biological, medical imaging or
signal detection only indirect observations are available for statistical inference, and these  
 problems are called inverse problems  in the (statistical) literature.
  In the case of medical imaging, a well-known example is Positron Emission Tomography. Here, the connection
between the 'true' image and the observations involves the Radon transform [see,  for example, \cite{MR1790012}]. Other typical examples
are the  reconstruction of biological or astronomical images, where the connection between
the true image and the observable image is - at least in a first approximation - given by convolution-type operators
[see,  for example, \cite{0266-5611-11-4-003} or \cite{MR2565572}]. Whereas in 
these models the data is in general described in a regression framework, similar (de-)convolution
problems arise in density estimation from indirect observations [see \cite{MR1224414} for an early reference]. The corresponding (multivariate)
statistical model for density deconvolution 
is defined by 
\beq{g10} Y_i=Z_i+\ve_i,\quad i=1,\hdots,n,
\eeq
where  $(Z_1,\ve_1)  , \ldots , (Z_n,\ve_n) \in\bb{R}^d \times \bb{R}^d $ are independent identically distributed random variables and  
the noise terms $\ve_1 , \ldots , \ve_n$ are are  also independent  the of the random variables $Z_1, \ldots , Z_n$. We assume that the 
density  $f_\ve$  of the errors  $\ve_i$ is known  and are interested in properties of the density $f$
of the random variables  $Z_i$ based on the sample $ \{Y_1 ,\ldots , Y_n \} $. In terms of densities, 
model \eqref{g10} can be rewritten as
\[
 g=f*f_\ve,
\]
where $g$ denotes the density of $Y_1$.  Density
  estimators can be  constructed and investigated  similarly to the regression case (see the references in the next paragraph), 
  and in this paper we are interested in describing qualitative features of the density $f$ using the sample $\{Y_1,\hdots,Y_n\}$. In particular we will
 develop a method for simultaneous detection of regions of monotonicity of the density
$f$ at a controlled level and construct a procedure  for the detection of the modes of $f$.
To our best knowledge multivariate  problems of this type have not been investigated so far in the literature.

On the other hand there exists a wide range of literature concerning statistical inference in the 
univariate deconvolution model. A Fourier-based  estimate of the density $f$  using a damping factor for large frequencies 
 was introduced in
\cite{MR1224414}, whereas  \cite{MR1765627} estimate $f$ with a wavelet-based deconvolution density estimator
[see also  \cite{MR1700237} for  a nonparametric estimator for the corresponding distribution function 
or \cite{MR2132729} for  a plug-in  estimator of $f$  based on estimation of 
a scale parameter for the noise level].   \cite{RSSB:RSSB599} develop confidence bands for deconvolution kernel
density estimators, while 
 minimax rates for this estimation problem can be found in  \cite{MR997599}  and  \cite{MR1126324}.   \cite{romano1988} and \cite{grundhall1995}
 point out that the detection of regions of monotonicity and of the modes of a density
is a more complex problem and  
 \cite{MR1126324} shows that the  minimax rate for estimating the derivative over a H\"older-$\beta$-class ($\beta\geq2$)
in the univariate setting $d=1$ is given by $n^{-(\beta-1)/(2\beta+2r+1)}$, where $r>0$ denotes the order of polynomial decay
of the Fourier transform of the error density $f_\ve$. 
 \cite{MR2757197} develop  a test for the number of modes of a univariate density and 
\cite{MR2493014} proposes  a local test for monotonicity for a fixed interval. More recently 
\cite{MR3113812} discuss  multiscale tests for qualitative features of a univariate density which provide uniform
confidence statements about shape constraints such as local monotonicity properties.
Little research has been done regarding multivariate deconvolution problems. Recent references for density estimation are e.g. 
\cite{MR3088382} using kernel density estimators and \cite{bay} for a Bayesian approach in the case of an unknown error 
distribution with  replicated proxies available. Hypothesis testing in deconvolution is investigated in \cite{MR2292917}
and \cite{MR2421946}.  
 
In the present paper we will develop  a multiscale method for  simultaneous identification of regions of 
monotonicity of the multivariate density $f$ in the  deconvolution model \eqref{g10}. 
Our approach is based on simultaneous  local tests of the directional derivatives
of the density $f$ for a significant deviation from zero for ``various'' directions and locations. 
In Section \ref{Sec3} we present a  Fourier based method  for the  construction
of local tests, which will be used for the inference about the monotonicity properties of the density $f$. 
Roughly speaking, we propose a multiscale test investigating the sign of the derivatives of the density $f$ in different  locations 
and directions and on different scales.
Section \ref{sec3a}  is devoted to asymptotic properties, which can be used  to obtain 
 a multiscale test for simultaneous confidence statements about the density. 
 Moreover, we also   propose a method
for the detection and localization of the modes. The finite sample properties of the method are discussed
in Section \ref{Sec4} and all
proofs  are deferred to Sections \ref{Sec5} and \ref{techres}, while Section \ref{sec7} contains two technical results.

\section{Multiscale inference in multivariate deconvolution}\label{Sec3}
\def\theequation{2.\arabic{equation}}
\setcounter{equation}{0}

Let $\partial_s$ denote the directional derivative in the direction of $s\in S^{d-1}=\{s\in\bb{R}^d\,| \,\|s\|=1\}$
and $\phi:\bb{R}^d\rightarrow\bb{R}_{\geq0}$ be a sufficiently smooth kernel (i.e. $\|\phi\|_{L^1(\bb{R}^d)}=1$) with compact support in $[-1,1]^d$. Define
\[\phi_{t,h}(.)=h^{-d}\phi\big(\tfrac{.-t}{h}\big)\quad\tn{for }t\in[0,1]^d,h>0.\]
For the description of the  local monotonicity properties of the function $f$  we introduce the integral
\beq{a5}-\int_{\bb{R}^d} \partial_sf(x)\phi_{t,h}(x)\tn{d}x. \eeq
If this expression is, say, negative, we can conclude that the
 derivative of $f$ in direction $s$ has to be strictly larger than zero on a subset of positive Lebesgue measure of the cube
$[t_1-h,t_1+h]\times\hdots\times[t_d-h,t_d+h]$. 

Statistical inference regarding the monotonicity properties 
of $f$ can then be performed by testing simultaneously several hypotheses of the form 
\begin{align}
\label{t5neu}H_{0,incr}^{s^j,t^j,h_j}: -\int_{\bb{R}^d} \partial_{s^j}f(x)\phi_{t^j,h_j}(x)\tn{d}x\geq0 ~~   
&\mbox{versus}  ~ ~H_{1,incr}^{s^j,t^j,h_j}:-\int_{\bb{R}^d} \partial_{s^j}f(x)\phi_{t^j,h_j}(x)\tn{d}x<0~
\end{align}
and
\begin{align}
\label{t5neu1}H_{0,decr}^{s^j,t^j,h_j}: -\int_{\bb{R}^d} \partial_{s^j}f(x)\phi_{t^j,h_j}(x)\tn{d}x\leq0 ~~   
&\mbox{versus}  ~  ~H_{1,decr}^{s^j,t^j,h_j}:-\int_{\bb{R}^d} \partial_{s^j}f(x)\phi_{t^j,h_j}(x)\tn{d}x>0~,
\end{align}
where $ (s^1,t^1,h_1), \ldots ,  (s^p,t^p,h_p) $ are given triples of directions, locations and scaling factors.

This method allows for a global understanding of the shape of the density $f$. A particular feature of the proposed 
method consists in the fact that by conducting
 formal statistical tests
the multiple level can be controlled (see Theorem \ref{t10}).

For example, simultaneous tests for 
hypotheses of the form \eqref{t5neu} and \eqref{t5neu1} can be used  to obtain a graphical representation 
of the local monotonicity behavior of the density as displayed in Figure \ref{g25}
for a bivariate density. 
The displayed  map is based on tests for the  hypotheses \eqref{t5neu}  
 for a fixed scale  $h_0$ and different locations and directions  $ (s^1,t^1), \ldots ,  (s^p,t^p) $ (here taken as  the vertices of an equidistant grid
 and  four equidistant directions on $S^1$). Note that we are investigating here a symmetric set of triples,
that is, for every location $t^j$  both the triple $(s^j,t^j,h_0)$ and $(-s^j,t^j,h_0)$ are considered.
Thus, as $H_{0,incr}^{s^j,t^j,h_0}=H_{0,decr}^{-s^j,t^j,h_0}$,
it is sufficient to investigate only hypotheses of the form \eqref{t5neu} in this setting. The figure shows the results of the tests for the different hypotheses in \eqref{t5neu}.
 An arrow in a direction  $s^j$ at a location $t^j$ represents a rejection of the corresponding hypothesis $H_{0,incr}^{s^j,t^j,h_0}$
 and provides therefore an indication of a positive directional derivative 
 of $f$ in direction $s^j$ at the  location $t^j$.
\begin{figure}
\begin{center}
\includegraphics[width=0.4\textwidth]{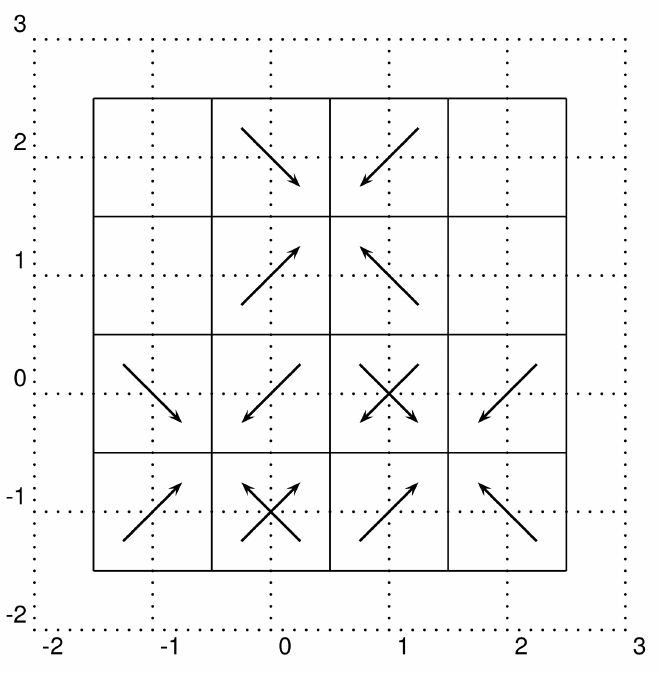}
\end{center}
\caption{\it Example of a global map for monotonicity of a bivariate density.} 
\label{g25}\end{figure}
For a detailed description of the settings used to provide Figure \ref{g25} and an analysis of the results we refer to Section \ref{g27}. 

If one is interested in specific shape constraints of the density, say in a test for a mode (local maximum)
at a given point $x^0$, inference 
can be conducted investigating the hypotheses 
\begin{align}
\label{t5einmode}
H_{0,decr}^{s^j,t^j,h_0}~~~  &\mbox{versus}  ~ ~ ~H_{1,decr}^{s^j,t^j,h_0}
\end{align}
for different pairs $ (t^1,s^1), \ldots ,  (t^p,s^p) $, where $t^1, \ldots ,  t^p$ are points in a neighborhood of $x^0$ 
on the lines $\{ x^0 + \lambda s^j | \lambda > 0\} $ ($j=1,\ldots ,p$), respectively (of course, on could additionally use different scales here). \\

Throughout this paper   we will  assume that all  partial derivatives $\partial_sf$ of the density $f$
are  uniformly bounded, such that the estimated quantity \eqref{a5} is bounded by a constant which  does not depend on 
the triple $(s,t,h)$. Using integration by parts, Plancherel's  identity and the convolution theorem, we get
\begin{align}\label{h0} -\int_{\bb{R}^d} \partial_sf(x)\phi_{t,h}(x)\tn{d}x &=  \int_{\bb{R}^d}f(x)\partial_s\phi_{t,h}(x)\tn{d}x\\ \nonumber
&=  \frac{1}{(2\pi)^d}\int_{\bb{R}^d}\ca{F}(f)(y)\ol{\ca{F}(\partial_s\phi_{t,h})}(y)\tn{d}y\\ \nonumber
&=  \frac{1}{(2\pi)^d}\int_{\bb{R}^d}\ca{F}(g)(y)\ol{\bigg(\frac{\ca{F}(\partial_s\phi_{t,h})}{\ol{\ca{F}(f_\ve)}}\bigg)}(y)\tn{d}y\\ \nonumber
&=  \int_{\bb{R}^d}g(x)\ca{F}^{-1}{\bigg(\frac{\ca{F}(\partial_s\phi_{t,h})}{\ol{\ca{F}(f_\ve)}}\bigg)}(x)\tn{d}x.
\end{align}
Here, 
\begin{eqnarray*} 
\ca{F}(f)(y) &=&\int_{\bb{R}^d}e^{-iy.x}f(x)\tn{d}x, \\
\ca{F}^{-1}(f)(x)&=& \frac{1}{(2\pi)^d}\int_{\bb{R}^d}e^{ix.y}f(y)\tn{d}y\quad\big(x,y\in\bb{R}^d\big) 
\end{eqnarray*} 
 denote the Fourier transform and its inverse,
respectively, $\overline{z}$ is the complex conjugate of $z\in\mathbb{C}$ and $x.y$ stands for
the standard inner product of  $x,y\in\bb{R}^d$.

For the  construction of tests for the hypotheses in \eqref{t5neu}  and \eqref{t5neu1}    
 we define the statistic
\begin{equation}  \label{tsth}
T_{s,t,h}^n=\frac{1}{n}\sum_{i=1}^nF_{s,t,h}(Y_i),
\end{equation}
where
\begin{equation}  \label{fsth}
F_{s,t,h}(Y_i) = \ca{F}^{-1}  \Big(\frac{\ca{F}(\partial_s\phi_{t,h})}{\ol{\ca{F}(f_\ve)}}\Big)(Y_i).
\end{equation}
Because (by \eqref{h0})
\[\bb{E}(T_{s,t,h}^n)=-\int_{\bb{R}^d} \partial_sf(x)\phi_{t,h}(x)\tn{d}x,\]
it follows that $T_{s,t,h}^n$ is a reasonable estimate of the quantity defined in \eqref{a5}, and hence the statistics 
$T_{s,t,h}^n$ define the main tool to study qualitative features of the density $f$. Inference 
on  local monotonicity  of  the density  $f$  will then be  based on tests rejecting the hypotheses $H_{0,incr}^{s,t,h}$
for  small  values of the corresponding statistic $T_{s,t,h}^n$ and rejecting $H_{0,decr}^{s,t,h}$
for  large  values of $T_{s,t,h}^n$ for several 
directions $s\in S^{d-1}$, locations $t\in [0,1]^d$ and  scales $h>0$. The multiple level of these tests can be controlled 
by investigating the (asymptotic) maximum of appropriately normalized  statistics $T_{s,t,h}^n$ calculated over a certain set of
locations, directions and scales.

\section{Asymptotic properties}  \label{sec3a}
\def\theequation{3.\arabic{equation}}
\setcounter{equation}{0}

In this section we investigate the asymptotic properties of a  statistic which can be used to control the multiple level of the  tests introduced
in Section \ref{Sec3}.   To be precise, we consider the finite subset
\[
 \ca{T}_n:=\big\{(s^j,t^j,h_j)\,|\,j=1,\hdots,p\big\}\subseteq S^{d-1}\times[0,1]^d\times[h_{\tn{min}},h_{\tn{max}}]
\]
of cardinality $ p\leq n^K$ for the calculation of the maximum of appropriately standardized statistics
$T_{s,t,h}^n$, where  $K>1$ and  for some $\ve>0$
\begin{equation} \label{hminmax} 
h_{\tn{min}}\gtrsim n^{-1/d+\ve}  ~ \mbox{ and  } ~h_{\tn{max}}=o((\log(n)\log\log(n))^{-1}).
\end{equation}
Throughout this paper we will make frequent use of multi-index notation, 
where  $\bs{\alpha}=(\alpha_1,\hdots,\alpha_d)\in\bb{N}_0^d$  denotes  a  multi-index (written in bold),  $|\bs{\alpha}|=\alpha_1+\hdots+\alpha_d$
its ``length'', and for a sufficiently smooth  function  $f:\bb{R}^d\rightarrow\bb{R}$ and a multi-index $\bs{\alpha}$ we denote by
 $$\partial^{\bs{\alpha}}f(x)=\frac{\partial^{|\bs{\alpha}|}}
{\partial x_1^{\alpha_1}\cdot\hdots\cdot\partial x_d^{\alpha_d}}f(x) 
$$
its  partial derivative. 

Recall the definition of $F_{s,t,h}$ in \eqref{fsth}, to simplify the notation define  for a point  $(s^j,t^j,h_j)  \in  \ca{T}_n$ 
\begin{equation} \label{Fj} 
F_j = F_{s^j, t^j, h_j}
\end{equation}
and consider  the random variables 
\beq {labelXj1} \td{X}_{j}^{(1)}=\frac{\sqrt{\log(eh_{j}^{-d}})}{\log\log(e^eh_{j}^{-d})} \Big( \frac{h_{j}^{d/2+r+1}}{\sqrt{n\hat{g}_n(t_j)}V_j} \Big |\sum_{i=1}^n
F_j(Y_i)-{n}\bb{E}(F_j(Y_1)) \Big |
-\sqrt{2\log(h_{j}^{-d})}\Big),
\eeq
where $\hat{g}_n$ is a density estimator of $g$ satisfying
\beq {dichte}
 \|g-\hat{g}_n\|_\infty=o(\log(n)^{-1})\quad\text{almost surely}
\eeq
(for example  a kernel density estimator  as considered in \cite{MR1955344}) and  
\beq  {Vj}
 V_j=h_j^{d/2+r+1}\|F_{s^j,t^j,h_j}\|_{L^2(\bb{R}^d)}.
\eeq
The quantity $V_j$ is well-defined under the assumptions presented below (see Lemma \ref{a11}  for details).

Note that   the boundary of the hypotheses
$H_{0,incr}^{s^j,t^j,h_j}$ and $H_{0,decr}^{s^j,t^j,h_j}$  in
\eqref{t5neu} and \eqref{t5neu1} is defined by   $\int_{\bb{R}^d} \partial_{s^j}f(x)\phi_{t^j,h_j}(x)\tn{d}x=0$
and in this case we have
$$
{1 \over \sqrt n} \td{X}_{j}^{(1)} = \frac{\sqrt{\log(eh_{j}^{-d}})}{\log\log(e^eh_{j}^{-d})} \Big( \frac{h_{j}^{d/2+r+1}}{\sqrt{\hat{g}_n(t_j)}V_j}  \big |T^n_{s^j,t^j,h_j} \big | -
\frac{\sqrt{2\log(h_{j}^{-d})}}{\sqrt n} \Big).
$$
Consequently,  we will investigate the asymptotic properties of 
$ \max_{1\leq j\leq p} \td{X}_{j}^{(1)}$ in the following discussion. 
For this purpose we make the following assumptions.

 \begin{ass}\label{q2}\rm{
  Assume that the density $g$  is Lipschitz continuous  and locally bounded from
  below, i.e.
  \[
   g(x)\geq c>0\text{ for all }x\in[0,1]^d.
  \]}
  \end{ass}

 \begin{ass}\label{3}\rm{
We assume a polynomial decay of the Fourier transform of the
 error density $f_\ve$, i.e. that there exist constants  $r>0$ for $d\geq 2$ resp. $r>\frac{1}{2}$ for $d=1$ and $0<C_u<C_o$ such that
 \[
  C_u\big(1+\|y\|^{2}\big)^{-r/2}\leq |\ca{F}(f_\ve)(y)|\leq C_o\big(1+\|y\|^{2}\big)^{-r/2}.
 \]
Furthermore, let
\[
 \sum_{j=1}^{\lceil (d+1)/2\rceil}(1+\|y\|^2)^{j/2}\Big|\frac{\partial^{j}}{\partial y_l^{j}}{\ca{F}(f_\ve)}(y)\Big|\leq C_o(1+\|y\|^2
 )^{-{r}/{2}}
\]
for all $l=1,\hdots,d$.}
\end{ass}
Note that as a direct consequence of Assumption \ref{q2} $g$ is bounded from above and that there exists a constant
$\delta>0$ such that $ g(x)\geq \frac{c}{2}>0\text{ for all }x\in
  [-\delta,1+\delta]^d$.
Assumption \ref{3} can be seen as a multivariate generalization of the classical assumptions on the decay of the Fourier transform
of the error density in the ordinary smooth case (see  e.g.  \cite{MR3113812}, Assumption 2). We also note that this 
assumption defines a
 mildly ill-posed situation (see \cite{MR2421946}).
 The next  assumptions refer to the kernel $\phi$ and are required for some technical arguments.
\begin{ass}\label{a8}\rm{
Let $\|\partial_s\phi\|_{L^2(\bb{R}^d)}\neq 0$ for all $s\in S^{d-1}$ and  assume that $\partial^{\bs{\beta}}\phi$ exists in $[-1,1]^d$  and is continuous for all $|\bs{\beta}|\leq \lceil r+2\rceil$, where $r$ is the constant from Assumption \ref{3}. We assume further that
for  some $\delta>0$ the inequality
\[
  \int_{\bb{R}^d}\big(1+\|y\|^2\big)^{r+({d+\delta})/{2}}\Big|\frac{\partial^{m}}{\partial y_l^{m}}{\ca{F}(\partial_{e^k}\phi)(y)}
{}\Big|^2\tn{d}y<\infty
\]
holds for all $k,l=1,\hdots,d$ and  $m=0,\hdots,\lceil (d+1)/2\rceil$, where $e^k,\;k=1,\hdots,d,$ denotes the $k$th unit vector of
$\bb{R}^d$.}
\end{ass}
As 
\[
 \Big|\frac{\partial^{m}}{\partial y_l^{m}}{\ca{F}(\partial_{s}\phi)(y)}
{}\Big|^2=\Big|\sum_{k=1}^ds_k\frac{\partial^{m}}{\partial y_l^{m}}{\ca{F}(\partial_{e^k}\phi)(y)}
{}\Big|^2\leq C \sum_{k=1}^d\Big|\frac{\partial^{m}}{\partial y_l^{m}}{\ca{F}(\partial_{e^k}\phi)(y)}
{}\Big|^2
\]
for  all $s\in S^{d-1}$ and some constant $C>0$ that only depends on $d $, Assumption \ref{a8} yields a uniform upper bound for the integral
\[
  \int_{\bb{R}^d}\big(1+\|y\|^2\big)^{r+({d+\delta})/{2}}\Big|\frac{\partial^{m}}{\partial y_l^{m}}{\ca{F}(\partial_{s}\phi)(y)}
{}\Big|^2\tn{d}y
\]
for all $s\in S^{d-1}$.  

Recall the definition of $\td{X}_{j}^{(1)}$   in \eqref{labelXj1}  and define the vector 
$\td{X}^{(1)}=(\td{X}_{1}^{(1)},\hdots,\td{X}_{p}^{(1)})^\top$. Our first main result provides a uniform approximation of the 
probabilities $\bb{P}( \td{X}^{(1)}\in A ) $ by  the probabilities  $\bb{P} (\td{X}\in A )$
 for every half-open hyperrectangle $A$, where the components of the vector $\td{X}=(\td{X}_{1},\hdots,\td{X}_{p})^\top$
 are defined by
\beq {xtilde}
\td{X}_{j}=\frac{\sqrt{\log(eh_{j}^{-d}})}{\log\log(e^eh_{j}^{-d})}\Big(h_{j}^{d/2+r+1}\frac{|\int_{\bb{R}^d}F_{j}(x)\tn{d}B_x|}
{V_j}-\sqrt{2\log(h_{j}^{-d})}\Big)
\eeq
$(j=1,\hdots,p)$, and  $(B_x)_{x\in\bb{R}^d} $ is a standard $d$-variate Brownian motion.
\begin{satz}\label{a9}
Let $\ca{A}$ denote the set $\ca{A}:=\{(-\infty,a_1]\times\hdots\times(-\infty,a_p]\,|\,a_1,\hdots,a_p\in\bb{R}\}.$
Then,
\begin{align}\begin{split}\label{a10}
\sup_{A\in\ca{A}}\big|\bb{P}\big( \td{X}^{(1)}\in A\big)-\bb{P}\big(\td{X}\in A\big)\big|=
o(1)\quad\text{for }n\rightarrow\infty.
 \end{split}\end{align}
Furthermore, the random variable $\max_{1\leq j\leq p} \td{X}_{j}$ is almost surely bounded uniformly with respect to $n$.
\end{satz}

Theorem \ref{a9} will be used to control the multiple level of statistical tests for the hypotheses of the form
\eqref{t5neu} and \eqref{t5neu1}. To this end, let $\alpha\in(0,1)$ and denote by $\kappa_n(\alpha)$ the smallest number such that
\beq{t11}
 \bb{P}\Big(\max_{1\leq j\leq p}\td{X}_j\leq \kappa_n(\alpha)\Big)\geq 1-\alpha.
\eeq
By  Theorem \ref{a9}, $\kappa_n(\alpha)$ is bounded uniformly with respect to $n$ and $\alpha$. 
The $j$th hypothesis in \eqref{t5neu} is rejected, whenever
\beq{t8}n^{-1}\sum_{i=1}^nF_j(Y_i) <-\kappa^j_n(\alpha),\eeq where
\beq{t11j}
 \kappa^j_n(\alpha)=\frac{\sqrt{\hat{g}_n(t_j)}V_j}{\sqrt{n}}h_j^{-d/2-r-1}\Big(\frac{\log\log(e^eh_{j}^{-d})}{\sqrt{\log(eh_{j}^{-d}})}\kappa_n(\alpha)
+\sqrt{2\log(h_{j}^{-d})}\Big).
\eeq
Similarly, the $j$th hypothesis in \eqref{t5neu1} is rejected, whenever
\beq{t81}n^{-1}\sum_{i=1}^nF_j(Y_i) >\kappa^j_n(\alpha).\eeq
\begin{satz}\label{t10}
Assume that the tests  \eqref{t8} and \eqref{t81} for the hypotheses \eqref{t5neu} and \eqref{t5neu1} are performed simultaneously for $j=1,\hdots,p$.
The  probability of at least one false rejection of any of the tests is asymptotically at most $\alpha$, that is
\begin{align*}
\bb{P}\Big(\exists j\in\{1,\hdots, p\}:\; n^{-1}|\sum_{i=1}^nF_j(Y_i)|>\kappa^j_n(\alpha)\Big)\leq\alpha+o(1)
\end{align*}
for $n\rightarrow \infty$.
\end{satz}

Next we introduce a method for the detection and localization of the modes of the density. The main idea is to conduct the local
tests for modality proposed in \eqref{t5einmode} for a set of candidate modes which does not assume any prior knowledge
about the density. To be precise,
  we assume the following condition on the set $\ca{T}_n$:
for any fixed $h$ and $s$ the set $\{t:(s,t,h)\in\ca{T}_n\}$ is an equidistant grid in $[0,1]^d$ with grid width $h$.
Furthermore, for any fixed $t$ and $h$ the set $\{s:(s,t,h)\in\ca{T}_n\}$ is a grid in $S^{d-1}$ with grid width
converging to zero with increasing sample size.

This  grid  is now used  as follows to check  if  a point $x^0\in(0,1)^d$ is  a  mode of $f$. 
Let $\ca{T}_n^{x^0}\subset\ca{T}_n$ be the set of all triples  $(s,t,h)\in\ca{T}_n$ such that
$ch\geq\|x^0-t\|\geq2\sqrt{d} h$ for some $c>2\sqrt{d}$ sufficiently large and $\tn{angle}(x^0-t,s)\rightarrow0$ for $n\rightarrow\infty$.
By the condition on $\ca{T}_n$ defined above, the set  $\ca{T}_n^{x^0}$ is nonempty for sufficiently large $n$.
We now use the local tests \eqref{t81} for the hypotheses  \eqref{t5einmode} and decide for a mode at the point $x^0$ if the null hypotheses  in \eqref{t5einmode}  are rejected for 
all  triples in $\ca{T}_n^{x^0}$.
Note that by choosing the test locations as the vertices of an equidistant grid
no prior knowledge about the location of $x^0$ has to be assumed. Theorem \ref{t14} below  states that the procedure
detects all modes of the density
with asymptotic probability one as $n\rightarrow\infty$.

\begin{satz}\label{t14}

Let $x^0\in(0,1)^d$ denote an arbitrary mode of the density $f$ and assume that there exist functions  $g_{x^0}:\mathbb{R}^d\rightarrow\mathbb{R}$,   $\td{f}_{x^0}:\mathbb{R}\rightarrow\mathbb{R}$
such that  the density $f$
 has a representation of the form
\beq {frep}
 f(x)\equiv (1+g_{x^0}(x))\td{f}_{x^0}(\|x-x^0\|)
\eeq
 (in a neighborhood of $x^0$), $g_{x^0}$ is differentiable in a neighborhood of the point $x^0$ such that  $g_{x^0}(x)=o(1)$ and
$\langle\nabla g_{x^0}(x),{e}\rangle=o(\|x-x^0\|)$ if   $x\rightarrow x^0$  for all
${e}\in\mathbb{R}^d$ with $\|{e}\|=1$. In addition, let $\td{f}_{x^0}$  be differentiable in a neighborhood
of the point $0$ with $\td{f}_{x^0}'(h)\leq-ch(1+o(1))$ for  $h\rightarrow0$.

If the set $$\big\{(s,t,h)\in\ca{T}_n: h\geq C \log(n)^{1/(d+2r+4)}n^{-1/(d+2r+4)}\big\}$$ for some $C>0$ sufficiently large is nonempty, then
the procedure described in the previous paragraph detects
 the mode  $x^0$  with asymptotic probability one as $n\rightarrow\infty.$
\end{satz}

The method to detect the modes of the density proposed in Theorem \ref{t14} proceeds in two steps: the verification of the presence of a mode with 
asymptotic probability one in the asymptotic regime presented above and its localization at the rate 
 $n^{-1/(d+2r+4)}$ (up to some logarithmic factor) given by the grid width.

\section{Finite sample properties}\label{Sec4}
\def\theequation{4.\arabic{equation}}
\setcounter{equation}{0}
In this section we illustrate the finite sample properties of the proposed multiscale inference.
The performance of the  test for modality at a given point
$x^0$ (see the hypotheses in \eqref{t5einmode}) and the dependence of its power on the bandwidth and the error variance  is investigated.
We also illustrate  how simultaneous tests for 
hypotheses of the form \eqref{t5neu} and \eqref{t5neu1} can be used to obtain a graphical representation of the local 
monotonicity properties of the density.

We consider two-dimensional densities, i.e. $d=2$. The density $f_\ve$ of the errors in model \eqref{g10} is given by a 
symmetric bivariate Laplacian with scale
parameter $\sigma>0$ which is defined through its characteristic function
\beq{g26}
 \ca{F}(f_\ve)(y_1,y_2)=\frac{1}{1+\frac{1}{2}\sigma^2(y_1^2+y_2^2)}
\eeq
for $(y_1,y_2)\in\bb{R}^2$ (cf. \cite{Kotz2001}, Chapter 5). This means that $r=2$ and straightforward calculations show that 
\beq{g1}
 F_{s,t,h}(x_1,x_2)= \ca{F}^{-1}\Big(\tfrac{\ca{F}(\partial_s\phi_{t,h})}{\ol{\ca{F}(f_\ve)}}\Big)(x_1,x_2)=\Big(\partial_s-
   \frac{\sigma^2}{2}\big(\partial_{e^1}^2\partial_s+\partial_{e^2}^2\partial_s\big)\Big)\phi_{t,h}(x_1,x_2)
\eeq
for $(x_1,x_2)\in\bb{R}^2$. The test function is chosen as
\[
 \phi(x_1,x_2)=c_2(1-x_1^4)(1-x_2^4)\mathbbm{1}\big\{|x_1|\leq 1,|x_2|\leq 1\big\},
\]
where $c_2$ defines the normalization constant, that is 
$$c_2=\big\|(1-x_1^4)(1-x_2^4)\mathbbm{1}\big\{|x_1|\leq 1,|x_2|\leq 1\big\}\big\|_{L^1(\bb{R}^d)}^{-1}$$ (note that $\phi$
is smooth within its support). Moreover, the integration by parts formula gives
\[
  -\int_{\bb{R}^2} \partial_sf(x)\phi_{t,h}(x)\tn{d}x =  \int_{\bb{R}^2}f(x)\partial_s\phi_{t,h}(x)\tn{d}x
\]
 as $\phi$ vanishes on the boundary of its support. Finally, by the representation \eqref{g1} we find that
the deconvolution kernel possesses all properties that are used for the proof of Theorem \ref{a9} and 
therefore Theorem \ref{a9} is also satisfied for the function $\phi$.\\

Throughout this section the nominal level is fixed as $\alpha=0.05$.

\subsection{A local test for modality}
In this section we investigate the performance of a local test for the existence of a mode (more precisely a local maximum)
at a given location $x^0$ which is defined by testing several hypotheses of the form \eqref{t5einmode} simultaneously.
Moreover, the influence of the choice 
of the different parameters on the power of the test is also investigated. To be precise, we conduct four tests 
for the hypotheses \eqref{t5einmode} with a fixed bandwidth $h=h_0$. The postulated mode is given by the point
$x^0 =(0,0)^\top$ and the four directions and locations are chosen as  $s^1=t^1=(1,0)^\top$,
$s^2=t^2=(0,1)^\top$, $s^3=t^3=(-1,0)^\top$ and $s^4=t^4=(0,-1)^\top$. We conclude that $f$ has 
a  local maximum at the point $x^0=(0,0) ^\top$, whenever all hypotheses 
\[
 H_{0,decr}^{s^j,t^j,h_0},~j=1,\hdots,4,
\]
are rejected, that is
\beq{modtest}
 T_{s^j,t^j,h_0}^n>\kappa_n^j(\alpha)  ~\mbox{ for all  } j=1,\hdots,4,  
 \eeq
 where $\kappa_n^j(\alpha) $ is defined by  \eqref{t11j}.
An illustration of the considered situation is provided in Figure \ref{f2}.
\begin{figure}
\begin{center}
\includegraphics[width=0.3\textwidth]{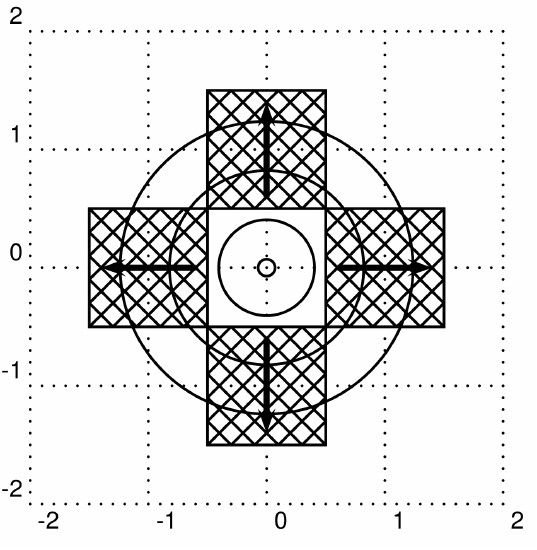}
\end{center}
\caption{\it Illustration of the four local tests for monotonicity used to define the test  \eqref{modtest} for $h_0=0.5$.
The crosshatched squares display the support of the functions $F_{s^j,t^j,h_0}$, $j=1,\hdots,4$, and the arrows the directional
vectors $s^j$, $j=1,\hdots,4$.} \label{f2}
\end{figure}
The quantiles $\kappa_n(0.05)$  defined in  \eqref{t11} are derived by $1000$ simulation runs based on normal
distributed
random vectors. In Table \ref{ta1} we display the normalized quantiles $\sqrt{n}\kappa_n^1(0.05)$ 
for the sample sizes $n=500,1000,4000$ 
observations and $h_0=0.5$. Here, the value of the parameter of the Laplacian error density has been chosen as
 $\sigma=0.075$.

\begin{table}[ht]
 \centering
\begin{tabular}{c| c }
$n$&  $\sqrt{n}\kappa_n^1(0.05)$  \\
\hline
500&0.039\\
1000&0.044\\
4000&0.041\\
\end{tabular}
\caption{\it Simulated quantiles $\sqrt{n}\kappa_n^1(0.05)$ of the test \eqref{modtest}. The density $f_\ve$ is defined
in \eqref{g26}.}\label{ta1}
\end{table}
 
The approximation of the level of the
 test for a mode at the point $x^0$  defined by \eqref{modtest} is investigated using  a uniform distribution on the square $[-2.5,2.5]^2$ for the density
 $f$. For power considerations
we sample  the $Z_i$ in model \eqref{g10} from a standard normal distribution. The results are displayed in the left part of 
Table \ref{ta5}. By its construction, the multiscale method is rather
conservative but nevertheless it is able to detect the mode with increasing sample size. In order to obtain a better 
approximation of the nominal level we propose 
a calibrated version of the test, where the quantiles are chosen such that the test keeps its nominal level  $\alpha=0.05$.
Note that this calibration does not require any knowledge about the unknown density $f$.
The simulated rejection probabilities are presented in the right part of Table \ref{ta5} for the parameters 
$h_0=0.5$ and  $\sigma=0.075$. We find that the calibrated test performs very well.

\begin{table}[h]
 \centering
\begin{tabular}{ c lc | c l c| c l }
$n$& level&power&level (cal.)&power (cal.) \\
\hline
500&0.3&39.4&4.2&74.7\\
1000&0.1&71.1&4.0&93.3\\
4000&0.4&99.9&3.1&100\\
\end{tabular}
\caption{\it Simulated level and power of the test \eqref{modtest}  for  a mode  at the point $x^0=(0,0)^\top$ of
 a $2$-dimensional density. The random variables $Z_i$ in model \eqref{g10} are standard normal distributed.
 Second and third column: test defined by \eqref{modtest};   fourth and  fifth  column:  test defined by \eqref{modtest}, where
 the quantiles $\kappa_n^j(\alpha)$ are replaced by calibrated quantiles.}
\label{ta5}
\end{table}

Next we fix the number of observations, that is $n=1000$,  the value of the parameter $\sigma=0.075$ and vary the bandwidth $h_0$ to
investigate its influence on the power of the test \eqref{modtest}. Recall that by the proposed choice of a Laplacian error density,
the deconvolution kernel has compact support in $[-1,1]^2$. Hence, by
dividing the bandwidth by 2 a fourth of the area is considered and (roughly) a fourth of the number of observations is used
 for the local test. Thus, we observe a decrease in power of the test for decreasing values of bandwidths 
which is illustrated in Table \ref{ta4}. 

\begin{table}[h]
 \centering
\begin{tabular}{ c lc | c l c| c l }
$h_0$& level&power&level (cal.)&power (cal.) \\
\hline
0.3&0.5&7.8&4.6&35.3\\
0.4&0.2&29.6&4.5&71.7\\
0.5&0.1&71.7&4.0&93.3\\
0.6&0.2&95.3&4.8&99.5\\
\end{tabular}
\caption{\it Dependence of the power of the test \eqref{modtest} for a mode at the point $x^0=(0,0)^\top$ on the bandwidth in the situation 
of Table \ref{ta5} where the number of observations is fixed to $n=1000$.
 Second and third column: test defined by \eqref{modtest};   fourth and  fifth  column:  test defined by \eqref{modtest}, where
 the quantiles $\kappa_n^j(\alpha)$ are replaced by calibrated quantiles.}\label{ta4}
\end{table}

We also investigate the influence of the scale parameter $\sigma$ on the power of the test \eqref{modtest}. To this end,
we fix the bandwidth as $h_0=0.5$ and the number of observations as $n=1000$ and vary the value of $\sigma$. The
results are shown 
in Table \ref{ta3} and we observe that an increase in the value of $\sigma$ decreases the power of the test.
On the other hand the power of the tests is very stable for small values of $\sigma$.

\begin{table}[h]
 \centering
\begin{tabular}{ c lc | c l c| c l }
$\sigma$& level&power&level (cal.)&power (cal.) \\
\hline
0.0 (direct setting)&0.4&77.7&4.7&94.1\\
0.075&0.1&71.7&4.0&93.3\\
0.15&0.2&71.1&3.6&92.8\\
0.3&0.4&62.3&3.8&87.2\\
1.0&0.3&31.4&4.5&59.4\\
\end{tabular}
\caption{\it Dependence of the power of the test \eqref{modtest} for a mode at the point $x^0=(0,0)^\top$  on the scale parameter in the situation 
considered in Table \ref{ta5} where the number of observations is fixed to $n=1000$.
 Second and third column: test defined by \eqref{modtest};   fourth and  fifth  column:  test defined by \eqref{modtest}, where
 the quantiles $\kappa_n^j(\alpha)$ are replaced by calibrated quantiles.}\label{ta3}
\end{table}

Next we investigate the influence of the shape of the modal region on the power of the test \eqref{modtest}. To this end, we fix 
the values of
$h_0=0.5$ and $\sigma=0.075$ and use
normal distributed random variables $Z_i$ with mean zero and non-diagonal covariance matrices 
  \renewcommand{\arraystretch}{0.75}
\beq{lp1}\Sigma_1=\Big( \;\begin{matrix} 0&0.5\\-1&1.5 \end{matrix} \;\Big)\tn{ and } ~~ 
\Sigma_2=\Big( \begin{matrix} -0.5&1\\-2&2.5 \end{matrix} \Big). \eeq

The simulated rejection probabilities are presented in Table \ref{ta6} and show that the absolute values of the eigenvalues of
the covariance matrix have an influence on the power of the test. In the case of $\ca{N}(0,\Sigma_1)$-distributed
random variables $Z_i$ (eigenvalues 0.5 and 1) the test 
performs better as for standard normal observations (with both eigenvalues equal to one). In the case of 
$\ca{N}(0,\Sigma_2)$-distributed random variables $Z_i$
 (eigenvalues 0.5 and 1.5) the test performs slightly worse than in the first case but still better as for 
standard normal observations due to the eigenvalue with absolute value smaller than one. We  note again the superiority of 
the calibrated test.
\renewcommand{\arraystretch}{1}
\begin{table}[h]
 \centering
\begin{tabular}{ c c  c|  c c  }
&\multicolumn{2}{c|}{$ \Sigma_1$}&\multicolumn{2}{c}{$ \Sigma_2$}\\
$n$ &power &power (cal.) &power &power (cal.)  \\
\hline
500&78.5&94.7&72.6&92.6\\
1000&96.7&99.3&96.5&98.9\\
4000&100&100&100&100\\
\end{tabular}
\caption{\it  Dependence of the power of the test \eqref{modtest} for a mode at the point $x^0=(0,0)^\top$ on the shape of the modal region. The random variables $Z_i$ are centered normal 
distributed with covariance matrices $\Sigma_1$ and $\Sigma_2$  given in \eqref{lp1}. Second and fourth column: test defined by \eqref{modtest}; 
third and  fifth  column:  test defined by \eqref{modtest}, where
 the quantiles $\kappa_n^j(\alpha)$ are replaced by calibrated quantiles.}\label{ta6}
\end{table}

We also investigate the influence of a (slight) misspecification of the position of the candidate mode on the power of 
the test \eqref{modtest} in  the situation considered in Table \ref{ta5} with  candidate mode $x^0=(0.2,0.2)^\top$. The results
are presented in Table \ref{t2}. We find that the slight misspecification of the position of the candidate mode affects
the power of the method only slightly.\\

\begin{table}[h]
 \centering
\begin{tabular}{c  c c}
&\multicolumn{2}{c}{$x^0=(0.2,0.2)^\top$}\\
$n$&power &power (cal.)   \\
\hline
500&34.9&70.8\\
1000&70.1&89.3\\
4000&99.9&100
\end{tabular}
\caption{\it Influence of a misspecification of the mode on the power of the  test \eqref{modtest}
for a mode at the point $x^0=(0.2,0.2)^\top$. The random variables $Z_i$ in model \eqref{g10} are standard normal distributed
and therefore the true mode is given by $(0,0)^\top$.
Second column: test defined by \eqref{modtest};  third column:  test defined by \eqref{modtest}, where
 the quantiles $\kappa_n^j(\alpha)$ are replaced by calibrated quantiles.}\label{t2}
\end{table}

Finally we consider a bimodal density and conduct simultaneously local tests for modality based on the hypotheses \eqref{t5einmode} 
for the candidate modes
$x^1=(0,0)^\top$ and $x^2=(3,0)^\top$. We conduct eight tests 
for the hypotheses \eqref{t5einmode}  for a fixed bandwidth $h=h_0=0.5$ with  $s^1=s^5=t^1=(1,0)^\top$,
$s^2=s^6=t^2=(0,1)^\top$, $s^3=s^7=t^3=(-1,0)^\top$, $s^4=s^8=t^4=(0,-1)^\top$ and $t^5=(4,0)^\top$, $t^6=(3,1)^\top$,
$t^7=(2,0)^\top$, $t^8=(3,-1)^\top $ and conclude that $f$ has 
a  local maximum in $x^1=(0,0) ^\top$ whenever all hypotheses 
\[
 H_{0,decr}^{s^j,t^j,h_0},~j=1,\hdots,4,
\]
are rejected, that is
\beq{modtest1}
 T_{s^j,t^j,h_0}^n>\kappa_n^j(\alpha)  ~\mbox{ for all  } j=1,\hdots,4  
 \eeq
  and  that $f$ has 
a  local maximum in $x^2=(3,0) ^\top$ whenever all hypotheses 
\[
 H_{0,decr}^{s^j,t^j,h_0},~j=5,\hdots,8,
\]
are rejected, that is
\beq{modtest2}
 T_{s^j,t^j,h_0}^n>\kappa_n^j(\alpha)  ~\mbox{ for all  } j=5,\hdots,8,  
 \eeq
 where the quantile $\kappa_n^j(\alpha) $ is defined by  \eqref{t11j}.
An illustration of the considered scales is provided in Figure \ref{f3}.
For the investigation of the approximation of the nominal level we
consider a uniform distribution on the rectangle $[-2.5,5.5]\times[-2.5,2.5]$ for the density $f$.
The  scaling factor in the Laplace density is given by $\sigma=0.075$.  For power
investigations we consider two bimodal densities given by  a uniform mixture of a  standard normal distribution
and a $\ca N((3,0)^\top,I)$ distribution (symmetric) and a uniform mixture of a
${\cal N}((0.0)^\top,1.2I)$ and a
${\cal N}((3.2,0.1)^\top,0.8I)$ distribution (asymmetric). The results for the calibrated version of the test 
are given in Table \ref{nid}.

\begin{figure}
\begin{center}
\includegraphics[width=0.45\textwidth]{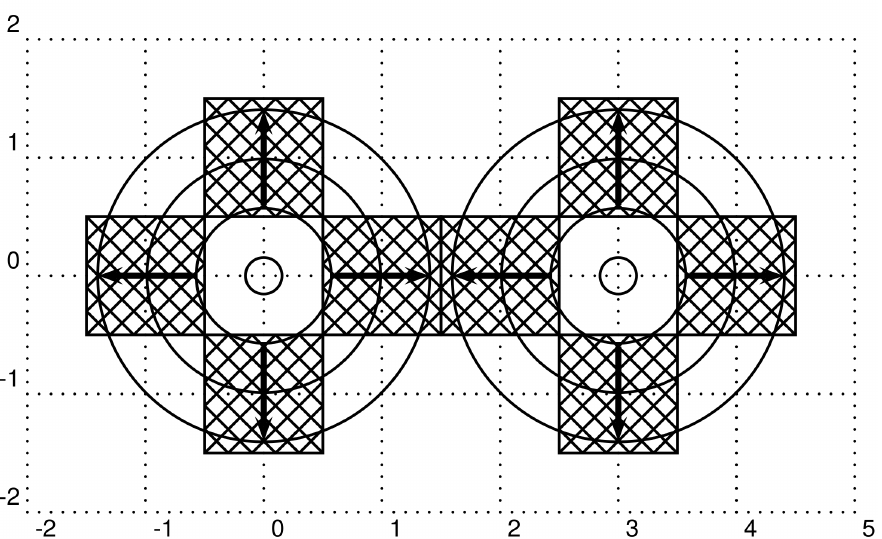}
\end{center}
\caption{\it  Illustration of the eight local tests for monotonicity used to create the tests \eqref{modtest1} and \eqref{modtest2}.
The crosshatched squares display the support of the functions $F_{s^j,t^j,h_0}$, $j=1,\hdots,8$, and the arrows the directional
vectors $s^j$, $j=1,\hdots,8$.} \label{f3}
\end{figure}

\begin{table}[h]
 \centering
\begin{tabular}{cccc|cc}
\multicolumn{2}{c}{}&\multicolumn{2}{c|}{Symmetric}&\multicolumn{2}{c}{Asymmetric}\\
$n$&level &power $x^1$ &power $x^2$ &power $x^1$&power $x^2$  \\
\hline
500&5.3&34.6&33.0&23.6&48.5\\
1000&5.2&48.7&49.9&39.0&72.9\\
4000&4.2&84.4&81.7&76.1&97.1\\
\end{tabular}
\caption{\it  Simulated level and power of the tests \eqref{modtest1} and \eqref{modtest2} for a mode at the
points $x^{1}=(0,0)^\top$ and $x^{2}=(3,0)^\top$, where
 the quantiles $\kappa_n^j(\alpha)$ are replaced by calibrated quantiles.
 The random variables $Z_i$ in model \eqref{g10} are  given by  a uniform mixture of a  standard normal distribution
and a $\ca N((3,0)^\top,I)$ distribution (symmetric) and a uniform mixture of a
${\cal N}((0.0)^\top,1.2I)$ and a
${\cal N}((3.2,0.1)^\top,0.8I)$ distribution (asymmetric).
}\label{nid}
\end{table}
 
We observe that in the symmetric case the test detects both modes with (roughly) the same power, whereas in the asymmetric
case the mode with smaller variance (even though there is a slight misspecification of its position) is detected more often.

A scatter plot of $n=4000$ observations from the convolution of the asymmetric bimodal density and a bivariate Laplace
distribution with 
scale parameter $\sigma=0.5$ is given in Figure \ref{g5}. Here, a look at the scatter plot does not give a hint
on the number of modes of the distribution. However, the test \eqref{modtest1},
where the quantiles $\kappa_n^j(\alpha)$ are replaced by calibrated quantiles,
is still able to detect
a mode at $(0,0)^\top$ in 48.4 percent of the repetitions and the test \eqref{modtest2} with calibrated quantiles detects a mode
in $(3,0)^\top$ in 81.4 percent of the repetitions. The simulated level for the calibrated quantiles is 4.1.
\begin{figure}
\begin{center}
\includegraphics[width=0.5\textwidth]{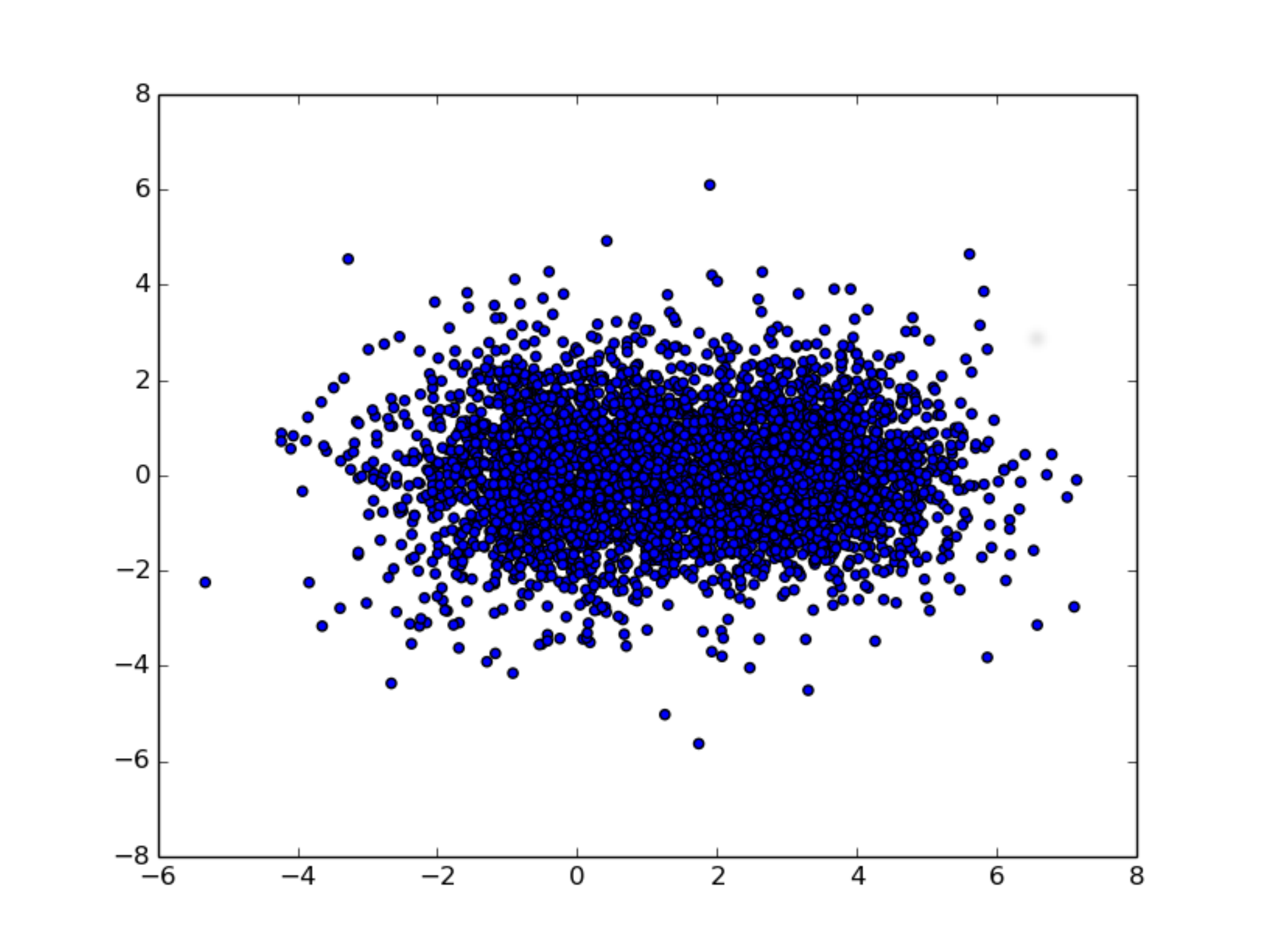}
\end{center}
\caption{\it $n=4000$ observations drawn from the convolution of a uniform mixture of a ${\cal N}((0.0)^\top,1.2I)$ and a
${\cal N}((3.2,0.1)^\top,0.8I)$ distribution and a bivariate Laplace distribution with scale parameter $\sigma=0.5$.} \label{g5}
\end{figure}

\subsection{Inference about local monotonicity  of a multivariate density}\label{g27}
The multiscale approach introduced in Section \ref{Sec3} can be used to obtain a graphical representation of the monotonicity behavior of a 
(bivariate) density. We construct a global map indicating monotonicity properties of the density $f$ by conducting the tests \eqref{t8} for the hypotheses \eqref{t5neu}
 for a fixed bandwidth of $h=0.5$.
The set of test locations $\ca{T}_t$ is defined as the set of vertices of an equidistant grid in the square 
$[-1,2]^2$ with  width $1$ and the set of test directions is given by
$\ca{T}_s=\{s^1=-s^3=\sqrt{2}^{-1}(1,1)^\top,s^2=-s^4=\sqrt{2}^{-1}(-1,1)^\top\}$.
The tests \eqref{t8} are conducted for every triple 
\[
 (s,t,h_0)\in\ca{T}_s\times\ca{T}_t\times\{h_0\}.
\]
The scaling factor for the Laplace density in the convolution model \eqref{g10} is given by
$\sigma=0.075$. We consider the tri-modal density with differently shaped modal regions displayed in Figure \ref{g24}. 
\begin{figure}
\begin{center}
\includegraphics[width=0.6\textwidth]{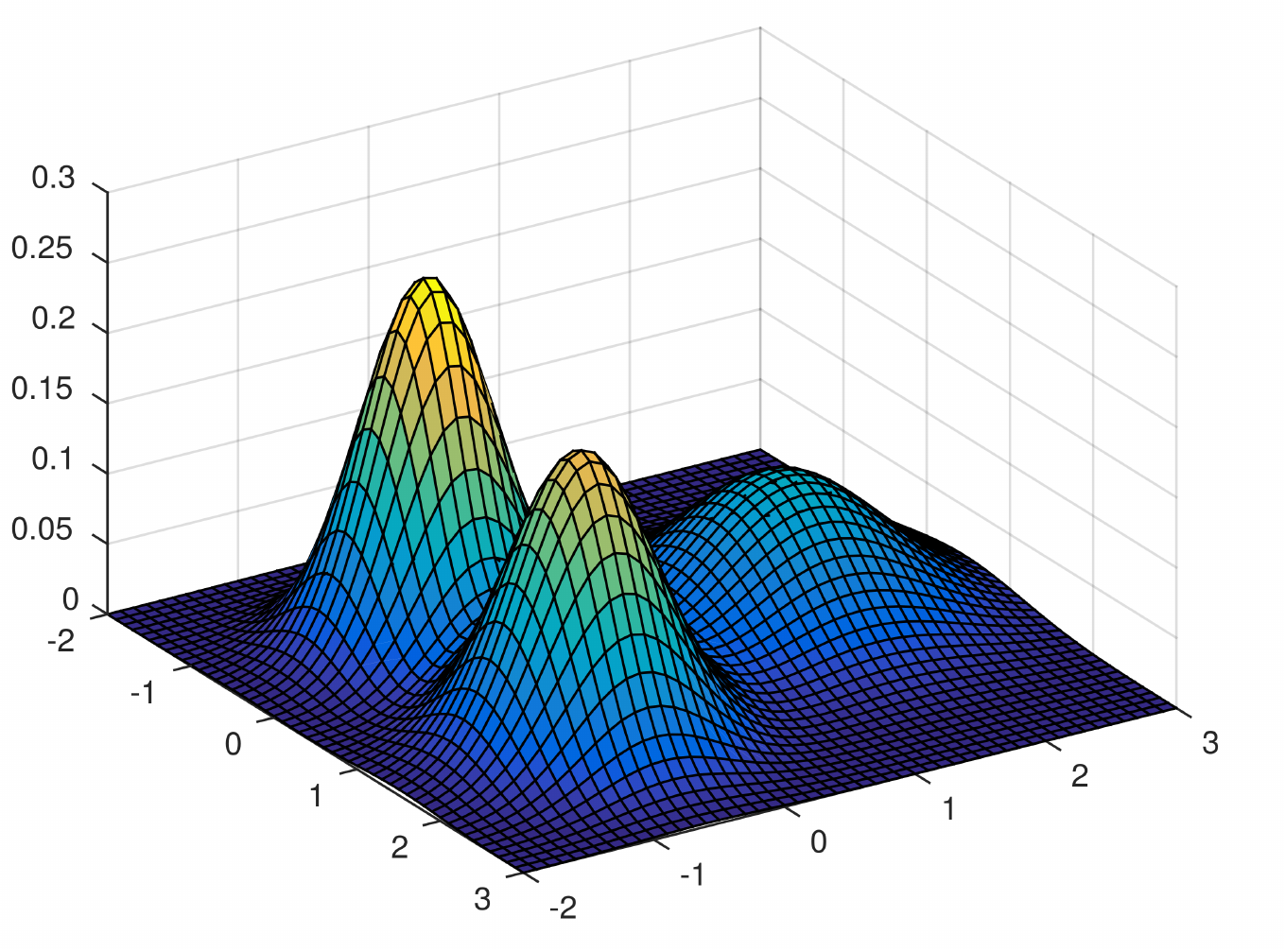}
\end{center}
\caption{\it The density of  a (uniform)   mixture  of  a ${\cal N}((-0.4,-0.57)^\top,0.2I),$
${\cal N}((1.5,-0.6)^\top,0.25I)$ and $ {\cal N}((0.45,1.6)^\top,0.5I)$ distribution.} \label{g24}\end{figure}

Figure \ref{g25} in Section \ref{Sec3} provides the graphical representation of the monotonicity behavior of the density $f$. Here, each arrow
at a location $t$ in direction $s$ displays a rejection of a hypothesis
 \eqref{t5neu}. The map indicates the existence of modes close to the  points $(-0.5,-0.5)^\top$, $(1.5,-0.5)^\top$ and 
$(0.5,1.5)^\top.$

\bigskip

{\bf Acknowledgements.}
This work has been supported in part by the
Collaborative Research Center ``Statistical modeling of nonlinear
dynamic processes'' (SFB 823, Project C1, C4) of the German Research Foundation (DFG).
The authors would like to thank Martina
Stein, who typed parts of this manuscript with considerable
technical expertise.

\setlength{\bibsep}{1pt}
\begin{small}
\bibliography{lit}
\end{small}

\section{Proof of Theorem \ref{a9}}\label{Sec5}
\def\theequation{5.\arabic{equation}}
\setcounter{equation}{0}
We split the proof of Theorem \ref{a9} in three parts. The first part is dedicated to several auxiliary results involving the deconvolution
kernel $F_{s,t,h}$. In the second part of the proof we show the approximation \eqref{a10}. Finally we conclude by proving the boundedness of the
limit distribution in the third part. \\
Throughout this section the symbols $\lesssim$ and $\gtrsim$ mean less or equal  and greater or equal, res\-pectively,
up to a multiplicative  
constant  independent of $n$ and $(s,t,h)$, and 
the symbol $|a_{s,t,h}|\asymp |b_{s,t,h}|$ means that $|a_{s,t,h}/b_{s,t,h}|$ is bounded from above and below by positive constants.

\subsection{Auxiliary results}
We begin with some basic transformations of the deconvolution kernel $F_{s,t,h}$. Recall that
\[
 F_{s,t,h}(.)=\ca{F}^{-1}\Big(\frac{\ca{F}(\partial_{s}\phi_{t,h})}{\ol{\ca{F}(f_\ve)}}\Big)(.)=
 {h^{-d-1}} \ca{F}^{-1}\Big(\frac{\int_{\bb{R}^d}e^{-iy.x}(\partial_{s}\phi)((x-t)/h)\tn{d}x}{\ol{\ca{F}(f_\ve)}(y)}\Big) (.)\]
 by  definition of the kernel $\phi_{t,h}$ and the Fourier transform. A substitution in the inner integral shows that
 \beq{a12}
   F_{s,t,h}(.)={h^{-1}}{} \ca{F}^{-1}\Big(\frac{e^{-iy.t}\ca{F}(\partial_{s}\phi)(hy)}{\ol{\ca{F}(f_\ve)}(y)}\Big) (.).
 \eeq
By the definition of the inverse Fourier transform and a substitution in the outer integral, we obtain
\beq{a14}
 F_{s,t,h}(x)=\frac{h^{-1}}{(2\pi)^{d}} \int_{\bb{R}^d}e^{ix.y}\frac{e^{-iy.t}\ca{F}(\partial_{s}\phi)(hy)}{\ol{\ca{F}(f_\ve)}(y)}\tn{d}y=
 \frac{h^{-d-1}}{(2\pi)^{d}} \int_{\bb{R}^d}e^{iy.\frac{x-t}{h}}\frac{\ca{F}(\partial_{s}\phi)(y)}{\ol{\ca{F}(f_\ve)}(y/h)} \tn{d}y.
\eeq
Furthermore, as $\partial_s\phi=\sum_{k=1}^ds_k\partial_{e^k}\phi$, where $e^k,\;k=1,\hdots,d,$ denotes the $k$th unit vector of
$\bb{R}^d$, we have
\[
 \ca{F}(\partial_s\phi)(y)=\sum_{k=1}^ds_kiy_k\ca{F}(\phi)(y),
\]
 where $i$ denotes the imaginary unit.
The following lemma presents some immediate consequences of the Assumptions \ref{3} and \ref{a8} made in Section \ref{sec3a}.
\begin{lemma}\label{20}~
Let $l\in\{1,\hdots,d\}$, $m\geq 2$ and $\td m=\lceil (d+1)/m\rceil$. It holds 

\begin{enumerate}
\item  $\quad\displaystyle S_s = \int_{\bb{R}^d}\big(1+\|{y}{}\|^2\big)^{{r}/{2}}\big|\ca{F}(\partial_{s}\phi)(y)\big|
\tn{d}y<\infty$ uniformly with respect to $s$;
 \item $\quad\displaystyle
 \int_{\bb{R}^d}\Big|
\frac{\partial^{\td m}}{\partial y_l^{\td m}}\Big(\frac{\ca{F}(\partial_{s}\phi)(y)}
{\ol{\ca{F}(f_\ve)}(y/h)}\Big)\Big|\tn{d}y\lesssim h^{-r}$.
\end{enumerate}

\end{lemma}
\begin{myproof}[\textbf{Proof of Lemma \ref{20}}]\hfill
 \textit{(i)}: An application of Cauchy-Schwartz's inequality yields  for any $\delta>0$
\bal
  S_s =& \int_{\bb{R}^d}\big(1+\|y\|^2\big)^{{r}/{2}+({d+\delta})/{4}}\big(1+\|y\|^2\big)^{-({d+\delta)/}{4}}\big|{\ca{F}(\partial_{s}\phi)(y)} \big|\tn{d}y
\\\leq& \Big(\int_{\bb{R}^d}\big(1+\|y\|^2\big)^{r+({d+\delta})/{2}}\big|{\ca{F}(\partial_{s}\phi)(y)} \big|^2\tn{d}y\Big)^{{1}/{2}}
\big\|\big(1+\|y\|^2\big)^{-({d+\delta})/{4}}\big\|_{L^2(\bb{R}^d)}.
\end{align*}
By Assumption \ref{a8}, there exists a constant $\delta>0$ such that the latter integral is bounded uniformly with 
respect to $s$. Hence, the assertion follows from the integrability of the function
$(1+\|y\|^2)^{-({d+\delta})/{2}}$.\\

 \textit{(ii)}:  By Leibniz's rule we have

\[
\Big|\frac{\partial^{\td m}}{\partial y_l^{\td m}}\Big(\frac{\ca{F}(\partial_{s}\phi)(y)}{\ol{\ca{F}(f_\ve)}(y/h)}\Big)\Big|\lesssim
\sum_{k=0}^{\td m}\Big|\frac{\partial^{\td m-k}}{\partial y_l^{\td m-k}}\ca{F}(\partial_{s}\phi)(y)
 \frac{\partial^{k}}{\partial y_l^{k}}\frac{1}{\ol{\ca{F}(f_\ve)}(y/h)}\Big|.
\]
Moreover, from Lemma \ref{s5} it follows that
\bal
  \Big| \frac{\partial^{k}}{\partial y_l^{k}}\frac{1}{\ol{\ca{F}(f_\ve)}(y/h)}\Big|\lesssim
  \sum_{(m_1,...,m_k)\in\ca{M}_k}\frac{1}{|\ol{\ca{F}(f_\ve)}(y/h)|^{m_1+\hdots+m_k+1}}h^{-k}\prod_{j=1}^k
  \Big|\Big(\frac{\partial^{j}}{\partial y_l^{j}}\ol{\ca{F}(f_\ve)}\Big)(y/h)\Big|^{m_j},
\end{align*}
where $\mathcal{M}_k$ is the set of all $k$-tuples of non-negative integers satisfying $\sum_{j=1}^k j m_j = k$. 
Assumption \ref{3} in Section \ref{sec3a} yields the estimates
\[
 \Big|\frac{\partial^{j}}{\partial y_l^{j}}\ol{\ca{F}(f_\ve)}(y)\Big|\lesssim \big(1+\|y\|^2
\big )^{-({r+j)/}{2}}\tn{ and }  \frac{1}{|\ol{\ca{F}(f_\ve)}(y)|}\lesssim \big(1+\|y\|^2\big)^{r/2}. \]
Thus, as $\sum_{j=1}^k j m_j = k$ for some $(m_1,\ldots,m_k) \in  \ca{M}_k$, we find
 \begin{align*}
 \Big| \frac{\partial^{k}}{\partial y_l^{k}}\frac{1}{\ol{\ca{F}(f_\ve)}(y/h)}\Big|
& \lesssim  h^{-k}\sum_{(m_1,...,m_k)\in\ca{M}_k} \big(1+\|\tfrac{y}{h}\|^2\big)^{(m_1+\hdots+m_k+1){r}{}/2}
\prod_{j=1}^k\big(1+\|\tfrac{y}{h}\|^2
 \big)^{-m_j({r+j})/{2}} \\
&\lesssim
   h^{-k}\sum_{(m_1,...,m_k)\in\ca{M}_k}\big(1+\|\tfrac{y}{h}\|^2\big)^{(m_1+\hdots+m_k+1){r}{}/2}
 \big(1+\|\tfrac{y}{h}\|^2\big)^{-(m_1+\hdots+m_k){r}{}/2 -k/{2}}\\
& \lesssim h^{-k}\big(1+\|\tfrac{y}{h}\|^2
 \big)^{(r-k)/2}.
 \end{align*}
Hence,
\[
\Big|\frac{\partial^{\td m}}{\partial y_l^{\td m}}\Big(\frac{\ca{F}(\partial_{s}\phi)(y)}{\ol{\ca{F}(f_\ve)}(y/h)}\Big)\Big|
\lesssim \sum_{k=0}^{\td m}h^{-k}\Big|\frac{\partial^{\td m-k}}{\partial y_l^{\td m-k}}\ca{F}(\partial_{s}\phi)(y)\Big|
\big(1+\|\tfrac{y}{h}\|^2 \big)^{(r-k)/2}.
\]
In the case $r\geq k$, the claim is now a direct consequence of the estimate 
\[
 h^{-k}\big(1+\|\tfrac{y}{h}\|^2\big)^{(r-k)/2}\lesssim h^{-r}(1+\|y\|^2)^{(r-k)/2},
\]
 similar arguments as given in proof of \textit{(i)}
and Assumption \ref{a8}.

If $r< k$ we divide the integration area into the ball $B_1(0)$ and its complement. For the integral
\[
 h^{-k}\int_{B_1(0)^C}\Big|\frac{\partial^{\td m-k}}{\partial y_l^{\td m-k}}\ca{F}(\partial_{s}\phi)(y)\Big|
\big(1+\|\tfrac{y}{h}\|^2 \big)^{(r-k)/2}\tn{d}y
\]
we have $h^{-k}\big(1+\|\tfrac{y}{h}\|^2\big)^{(r-k)/2}\lesssim h^{-r}$. Therefore, we can bound the integral over the complement of the
unit ball by the integral over $\bb{R}^d$ and proceed similarly to the first case.
It remains to consider the integral over the ball $B_1(0)$. To this end, notice that
\[
 h^{-k}\big(1+\|\tfrac{y}{h}\|^2 \big)^{(r-k)/2}\leq h^{-r}\|y\|^{r-k}.
\]
Hence, by the boundedness of $\frac{\partial^{\td m-k}}{\partial y_l^{\td m-k}}\ca{F}(\partial_{s}\phi)$ (which follows 
from the compactness of the
support of $\phi$)
 it remains to show that the integral
 \[
  \int_{B_1(0)}\|y\|^{r-k}\tn{d}y\lesssim \int_0^1 \rho^{d-1+r-k}\tn{d}\rho
 \]
is bounded, where  we used a polar coordinate transform to obtain the inequality. As $k\leq \lceil (d+1)/2\rceil$ and $r>0$,
the integral on the right hand side is obviously finite.
\end{myproof}

 Part \textit{(i)} of the following lemma shows that the constants  $V_1,\ldots,V_p$ defined in \eqref{Vj} are uniformly bounded from above
 and below.
 \begin{lemma}\label{a11}
 It holds
 \begin{enumerate}
  \item  $\|F_{s,t,h}\|_{L^2(\bb{R}^d)}\asymp h^{-{d}/{2}-r-1}$;
  \item $\big\|F_{s,t,h}\|x-t\|\big\|_{L^2(\bb{R}^d)}\lesssim h^{-{d}/{2}-r}$;
    \item $\|F_{s,t,h}F_{s',t',h'}\|_{L^1(\bb{R}^d)}\lesssim (hh')^{-{d}/{2}-r-1}$;
 \item $\big\|F_{s,t,h}F_{s',t',h'}\|x-t\|\|x-t'\|\big\|_{L^1(\bb{R}^d)}\lesssim (hh')^{-{d}/{2}-r}$.
\end{enumerate}

 \end{lemma}

\begin{proof}[\textbf{Proof of Lemma \ref{a11}}]   ~

  \textit{(i)}:  Using Plancherel's theorem and the representation \eqref{a12}, we obtain
 \begin{equation} \label{h1}
  \|F_{s,t,h}\|_{L^2(\bb{R}^d)}^2\asymp{h^{-2}}{}\Big\|\frac{e^{-iy.t}\ca{F}(\partial_{s}\phi)(h.)}{\ol{\ca{F}(f_\ve)}(.)}\Big\|_{L^2(\bb{R}^d)}^2
  =h^{-2}\ird\Big|\frac{\ca{F}(\partial_{s}\phi)(hy)}{\ol{\ca{F}(f_\ve)}(y)}\Big|^2\tn{d}y.
 \end{equation}

It now follows from  Assumption \ref{3} and a substitution that
 \[
  \|F_{s,t,h}\|_{L^2(\bb{R}^d)}^2\lesssim h^{-d-2r-2}\ird\big(1+\|y\|^2)^r\big|\ca{F}(\partial_{s}\phi)(y)\big|^2\tn{d}y,
 \]
 and the latter integral is bounded by Assumption \ref{a8} which concludes the proof of the upper bound.

 For the lower bound we find from (\ref{h1})   and Assumption \ref{3} that
\bal
 \|F_{s,t,h}\|_{L^2(\bb{R}^d)}^2 & 
 \gtrsim h^{-2}\ird \big(1+\|y\|^2\big)^{r}\big|\ca{F}(\partial_{s}\phi)(hy)\big|^2\tn{d}y  \\
& \gtrsim h^{-d-2}\ird \big(1+\|\tfrac{y}{h}\|^2\big)^{r}\big|\ca{F}(\partial_{s}\phi)(y)\big|^2\tn{d}y
 \gtrsim  h^{-d-2r-2}\int_{B_a(0)^C} \big|\ca{F}(\partial_{s}\phi)(y)\big|^2\tn{d}y
\end{align*}
for any constant $a>0$. Moreover,
\[
 \int_{B_a(0)^C} \big|\ca{F}(\partial_{s}\phi)(y)\big|^2\tn{d}y=\ird \big|\ca{F}(\partial_{s}\phi)(y)\big|^2\tn{d}y
 -\int_{B_a(0)} \big|\ca{F}(\partial_{s}\phi)(y)\big|^2\tn{d}y\gtrsim\|\partial_s\phi\|_{L^2(\bb{R}^d)}^2
\]
for a sufficiently small radius $a$ by the integrability of $|\ca{F}(\partial_{s}\phi)|^2$ (Assumption \ref{a8}) and Plancherel's theorem.
Furthermore, the mapping $s\mapsto\|\partial_s\phi\|_{L^2(\bb{R}^d)}$ is continuous such that by Assumption \ref{a8} $\|\partial_s\phi\|_{L^2(\bb{R}^d)}\geq c>0$
for a constant $c$ that does not depend on $s$.\\

 \textit{(ii)}: The representation \eqref{a14} and a substitution in the integral for the variable $x$ show
 \[
  \big\|F_{s,t,h}\|x-t\|\big\|_{L^2(\bb{R}^d)}^2=\frac{h^{-d}}{(2\pi)^{2d}}\ird\|x\|^2\Big|\ird e^{iy.x}\frac{\ca{F}(\partial_{s}\phi)(y)}{\ol{\ca{F}(f_\ve)}(y/h)}\tn{d}y\Big|^2\tn{d}x.
 \]
As $\|x\|^2=x_1^2+\hdots+x_d^2$, the differentiation rule for Fourier transforms yields
 \begin{align*}
  \big\|F_{s,t,h}\|x-t\|\big\|_{L^2(\bb{R}^d)}^2&=\frac{h^{-d}}{(2\pi)^{2d}}\sum_{k=1}^d\ird\Big|\ird e^{iy.x}\frac{\partial}{\partial y_k}\Big(\frac{\ca{F}(\partial_{s}\phi)(y)}{\ol{\ca{F}(f_\ve)}(y/h)}\Big)\tn{d}y\Big|^2\tn{d}x\\&
={h^{-d}}{}\sum_{k=1}^d\Big\|\ca{F}^{-1}\Big( \frac{\partial}{\partial y_k}\Big(\frac{\ca{F}(\partial_{s}\phi)(y)}{\ol{\ca{F}(f_\ve)}(y/h)}\Big)\Big)\Big\|_{L^2(\bb{R}^d)}^2  \\
&
\asymp{h^{-d}}{}\sum_{k=1}^d\Big\| \frac{\partial}{\partial y_k}
  \Big(\frac{\ca{F}(\partial_{s}\phi)(y)}{\ol{\ca{F}(f_\ve)}(y/h)}\Big)\Big\|_{L^2(\bb{R}^d)}^2,
 \end{align*}
 where the last identity follows from  Plancherel's theorem.
We now proceed similarly as in the proof of Lemma \ref{20} \textit{(ii)} and note that
 \[
  \frac{\partial}{\partial y_k}
  \frac{\ca{F}(\partial_{s}\phi)(y)}{\ol{\ca{F}(f_\ve)}(y/h)}=
  \frac{\partial}{\partial y_k}
  \ca{F}(\partial_{s}\phi)(y)\frac{1}{\ol{\ca{F}(f_\ve)}(y/h)}-
  \frac{\ca{F}(\partial_{s}\phi)(y)}{\big(\ol{\ca{F}(f_\ve)}(y/h)\big)^2}
   \frac{\partial}{\partial y_k}
 \big( \ol{\ca{F}(f_\ve)}(y/h)\big).
 \]
 An application of the Assumptions \ref{3} and \ref{a8} shows
 \[
  \Big\|  \frac{\partial}{\partial y_k}
 \ca{F}(\partial_{s}\phi)(y)\frac{1}{\ol{\ca{F}(f_\ve)}(y/h)}\Big\|_{L^2(\bb{R}^d)}^2\lesssim
  h^{-2r}\ird \Big| \frac{\partial}{\partial y_k}\ca{F}(\partial_{s}\phi)(y)\Big|^2\big(1+\|y\|^2\big)^r\tn{d}y\lesssim h^{-2r}.
 \]
Moreover, by Assumption \ref{3}, we have
\[
   \Big\|   \frac{\ca{F}(\partial_{s}\phi)(y)}{\big(\ol{\ca{F}(f_\ve)}(y/h)\big)^2}
   \frac{\partial}{\partial y_k}
 \big( \ol{\ca{F}(f_\ve)}(y/h)\big)\Big\|_{L^2(\bb{R}^d)}^2\lesssim
 h^{-2}\ird \big|\ca{F}(\partial_{s}\phi)(y) \big|^2\big(1+\|\tfrac{y}{h}\|^2\big)^{r-1}\tn{d}y.
\]
This concludes the proof for $r\geq 1$. For $r<1$ we
 split up the  area of integration into the ball $B_1(0)$ and its complement and find the required result for the integration over the complement using similar arguments as in     the proof of Lemma  \ref{20} \textit{(ii)}. For the integral over the unit ball  we also follow the line of arguments presented in
 the proof of Lemma  \ref{20} \textit{(ii)} which yields the required result provided that the integral on the right hand side
 of the inequality
   \[
  \int_{B_1(0)}\|y\|^{2r-2}\tn{d}y\lesssim \int_0^1 \rho^{d-1+2r-2}\tn{d}\rho
 \]
exists. This is the case for all $r>0$ if $d\geq2$   and all $r>\frac{1}{2}$ in the case $d=1.$\\

\textit{(iii)} and  \textit{(iv)}: These are  direct consequences of H\"older's inequality and \textit{(i)} resp. \textit{(ii)}.

\end{proof}

 The following Lemma will be used in the second part of the  proof of Theorem \ref{a9}.
 \begin{lemma}\label{2}For $1\leq j,k\leq p$ and  $m\geq 2$ we have for the function $F_j=F_{s^j,t^j,h_j}$ defined in \eqref{Fj}
\begin{enumerate}
 \item $|F_j(x)|\lesssim h_j^{-d-r-1}$ for all $x\in\bb{R}^d$;
 \item $\bb{E}(|F_j(Y_1)|^m)\lesssim h_{j}^{-(m-1)d-mr-m}$.

\end{enumerate}
\end{lemma}
 \begin{proof}[\textbf{Proof of Lemma \ref{2}}]   ~

 \textit{(i)}: Using the representation  \eqref{a14} and   Assumption \ref{3} it follows  that
 \[
  |F_j(x)|\lesssim h_j^{-d-1} \int_{\bb{R}^d}\Big|\frac{\ca{F}(\partial_{s^j}\phi)(y)}{\ol{\ca{F}(f_\ve)}(y/h_j)} \Big|\tn{d}y
  \lesssim h_j^{-d-r-1} \int_{\bb{R}^d}\big(1+\|y\|^2\big)^{r/2}\big|{\ca{F}(\partial_{s^j}\phi)(y)} \big|\tn{d}y=h_j^{-d-r-1}S_{s^j}.
 \]
The claim follows from the uniform boundedness of $S_{s^j}$ shown in Lemma \ref{20} \textit{(i)}.
\\

\textit{(ii)}: Using the representation \eqref{a14}, the boundedness of the density $g$ and a substitution we get
\bal
 \int_{\bb{R}^d}\big|F_j(x) \big|^mg(x)\tn{d}x&\lesssim h_j^{-md-m}\int_{\bb{R}^d}\Big|\int_{\bb{R}^d}e^{iy.\frac{x-t^j}{h_j}}\frac{\ca{F}(\partial_{s^j}\phi)(y)}
{\ol{\ca{F}(f_\ve)}(y/h_j)}\tn{d}y\Big|^m\tn{d}x\\
&=h_j^{-(m-1)d-m}\int_{\bb{R}^d}\Big|\int_{\bb{R}^d}e^{ix.y}\frac{\ca{F}(\partial_{s^j}\phi)(y)}
{\ol{\ca{F}(f_\ve)}(y/h_j)}\tn{d}y\Big|^m\tn{d}x.
\end{align*}
The proof will be completed showing the estimate
\[
 \int_{\bb{R}^d}\Big|\int_{\bb{R}^d}e^{ix.y}\frac{\ca{F}(\partial_{s^j}\phi)(y)}
{\ol{\ca{F}(f_\ve)}(y/h_j)}\tn{d}y\Big|^m\tn{d}x\lesssim h_j^{-mr}.
\]
For this purpose we decompose the domain of  integration  for the variable $x$ in two parts: the cube $[-\delta,\delta]^d$ for some $\delta>0$ and its
complement. For the integral with respect to the cube we use the upper bound $\int_{\bb{R}^d}\big|\frac{\ca{F}(\partial_{s^j}\phi)(y)}
{\ol{\ca{F}(f_\ve)}(y/h_j)}\big|\tn{d}y\lesssim h_j^{-r}$ provided in the proof of \textit{(i)} which yields the required result.

For the integral with respect to $([- \delta, \delta]^d)^C$ note that
\[
\int_{([-\delta,\delta]^d)^C}\Big|\int_{\bb{R}^d}e^{ix.y}\frac{\ca{F}(\partial_{s^j}\phi)(y)}
{\ol{\ca{F}(f_\ve)}(y/h_j)}\tn{d}y\Big|^m\tn{d}x \leq  \sum_{k=1}^d \sum_{l=1}^d  \int_{A_{k,l}}\Big|\int_{\bb{R}^d}e^{ix.y}\frac{\ca{F}(\partial_{s^j}\phi)(y)}
{\ol{\ca{F}(f_\ve)}(y/h_j)}\tn{d}y\Big|^m\tn{d}x~,
\]
where the sets $A_{k,l}$ are defined by
\bal
A_{k,l} =
 \big \{x\in\bb{R}^d\,|\,|x_k|>\delta,|x_l|\geq|x_{l'}|\text{ for all }l'\neq l\big\}.
\end{align*}
Now $\td m=\lceil (d+1)/m\rceil$ fold integration by parts yields
\[
 \Big|\int_{\bb{R}^d}e^{ix.y}\frac{\ca{F}(\partial_{s^j}\phi)(y)}{\ol{\ca{F}(f_\ve)}(y/h_j)}\tn{d}y\Big|^m
 =\frac{1}{|x_l|^{m\td m}}\Big|\int_{\bb{R}^d}e^{ix.y}\frac{\partial^{\td m}}{\partial y_l^{\td m}}\Big(
\frac{\ca{F}(\partial_{s^j}\phi)(y)}
{\ol{\ca{F}(f_\ve)}(y/h_j)}\Big)\tn{d}y\Big|^m,
\]
provided that $\frac{\partial^{\td m}}{\partial y_l^{\td m}}\big(\frac{\ca{F}(\partial_{s^j}\phi)(y)}
{\ol{\ca{F}(f_\ve)}(y/h_j)}\big)\in L^1(\bb{R}^d)$, which holds by Lemma \ref{20} \textit{(ii)}. A further application of
Lemma \ref{20} \textit{(ii)} shows that
\bal
\int_{A_{k,l}}\Big|\int_{\bb{R}^d}e^{ix.y}\frac{\ca{F}(\partial_{s^j}\phi)(y)}
{\ol{\ca{F}(f_\ve)}(y/h_j)}\tn{d}y\Big|^m\tn{d}x&\lesssim h_j^{-mr}\int_{[-\delta,\delta]^C}\frac{|x_l|^{d-1}}{|x_l|^{d+1}}\tn{d}x_l,
\end{align*}
as $|x_{l'}|\leq |x_l|$ for all $l'\neq l$ and $|x_l|>\delta$ in $A_{k,l}$.

\end{proof}

\subsection{Proof of the approximation \textbf{\eqref{a10}}}
For the consideration of the absolute values we introduce the set
$$\ca{T}_n':=\ca{T}_n\cup\{(-s,t,h)\,|\,(s,t,h)\in\ca{T}_n\}=:\{(s^j,t^j,h_j)\,|\,j=1,\hdots,2p\}$$ and denote by $\ca{A}'$
 the set of all hyperrectangles in $\bb{R}^{2p}$ of the form
\[
 A=\{w\in\bb{R}^{2p}\,|\,a_j\leq w_j\leq b_j\text{ for all }1\leq j\leq 2p\}
\]
for some $-\infty\leq a_j\leq b_j\leq\infty  \ (1\leq j\leq 2p)$.

We will show below in Section \ref{chern} that the random vectors $X_i=(X_{i,1},\hdots,X_{i,2p})^\top\in\bb{R}^{2p}$, $i=1,\hdots,n$, with
\[
  X_{i,j}=h_{j}^{d/2+r+1}\big({F_j(Y_i)-\bb{E}(F_j(Y_1))}\big){}\quad (i=1,\hdots,n,j=1,\hdots,2p)
\]
fulfill
\beq{a21}
 \sup_{A\in\ca{A}'}\Big|\bb{P}\Big( \frac{1}{\sqrt{n}}\sum_{i=1}^nX_i\in A\Big)-\bb{P}\Big(\frac{1}{\sqrt{n}}\sum_{i=1}^nY_i'\in A\Big)\Big|
 \lesssim \Big(\frac{h_{\tn{min}}^{-d}\log^7(n)}{n}\Big)^{1/6}+\Big(\frac{h_{\tn{min}}^{-d}\log^3(n)}{n^{1-2/q}}\Big)^{1/3}
\eeq
for any $q>0$, where $Y_1', \ldots, Y_n'$ are independent random vectors,     $Y_i'=(Y'_{i,1},\hdots,Y'_{i,2p})^\top\sim \mathcal{N}(0,\bb{E}( X_iX_i^\top))$,
 $i=1,\hdots,n$. Note that we have
\bal
 \frac{1}{\sqrt{n}}\sum_{i=1}^nY'_{i}&\sim N(0,\bb{E}( X_1X_1^\top)),
 \end{align*}
 where
\bal
\bb{E}( X_1X_1^\top)=
\Big({(h_jh_k)^{d/2+r+1}}{}\big(\bb{E}(F_j(Y_1)F_k(Y_1))
 -\bb{E}(F_j(Y_1))\bb{E}(F_k(Y_1))\big)\Big)_{1\leq j,k\leq 2p},
\end{align*}
as the random variables $X_1,\hdots,X_n$ are i.i.d. and $Y'_1,\hdots,Y'_n$ are independent.

Introduce a Gaussian process $(\td{B}(\Phi))_{\Phi\in L^\infty(\bb{R}^d)}$ indexed by $L^\infty(\bb{R}^d)$ as a process whose mean and
covariance functions are $0$ and
\beq{covstr}
 \ird \Phi_1(x)\Phi_2(x)g(x)\tn{d}x-\ird \Phi_1(x)g(x)\tn{d}x\ird \Phi_2(x)g(x)\tn{d}x,
 \eeq
 respectively. Hence,
there exists a version of $\td{B}(\Phi)$ such that
\[
  \frac{1}{\sqrt{n}}\sum_{i=1}^nY'_{i}=\big(h_1^{d/2+r+1}\td{B}(F_1),\hdots,h_{2p}^{d/2+r+1}\td{B}(F_{2p})\big)^\top.
\]
To derive an alternative representation of the process $\td B$ recall the definition of the isonormal process $(B(\Phi))_{\Phi\in L^2(\bb{R}^d)}$  as a Gaussian process whose mean and
covariance functions are $0$ and $\ird \Phi_1(x)\Phi_2(x)\tn{d}x$, respectively (see, e.g. \cite{kho}, Section 5.1). In particular, note that $(B(\mathbbm{1}_A))_{A\in\ca{B}(\bb{R}^d)}$ defines white noise, where
$\ca{B}(\bb{R}^d)$ denotes the Borel-$\sigma$-field on $\bb{R}^d$. Throughout this paper, we will use the notation
$B(\Phi)=\ird \Phi(x)\tn{d}B_x$.

There exists a version of the isonormal process such that $\td{B}(\Phi)=B(\Phi\sqrt{g})-\ird\Phi(x) g(x)\tn{d}x B(\sqrt{g})$
for $\Phi\in  L^\infty(\bb{R}^d)$ (one proves easily
 that $(B(\Phi\sqrt{g})-\ird\Phi(x){g}(x)\tn{d}x B(\sqrt{g}))_{\Phi\in L^\infty(\bb{R}^d)}$ defines a Gaussian process
with the covariance kernel \eqref{covstr}). Thus,
 \[
  \max_{1\leq j\leq 2p}\big|\td{B}(F_j)-B(F_j\sqrt{g})\big|=  \max_{1\leq j\leq 2p}\Big|\ird F_j(x){g(x)}\tn{d}x B(\sqrt{g})\Big|.
 \]
From (\ref{h0}) we have
\begin{equation} \label{h5}
\Big|\int_{\mathbb{R}^d} F_j(x)g(x)dx \Big| =   | \mathbb{E} [F_j(Y_1)]| = \Big | \int_{\mathbb{R}^d} \partial_s f(x) \phi_{t,h} (x) dx \Big | = O(1)
\end{equation}
uniformly with respect to $s,t,h$ (by assumption).
 Furthermore,
\[
 B(\sqrt{g})\sim N(0,\ird g(x)\tn{d}x)\sim N(0,1),
\]
 which implies that
\[
 \bb{E}\big(  \max_{1\leq j\leq 2p}h_j^{d/2+r+1} \big|\td{B}(F_j)-B(F_j\sqrt{g})\big| \big)\lesssim h_{\tn{max}}^{d/2+r+1}.
\]
An application of Markov's inequality finally proves
 \begin{equation} \label{h2}
   \max_{1\leq j\leq 2p}h_j^{d/2+r+1}\big|\td{B}(F_j)-B(F_j\sqrt{g})\big|=O_{\mathbb{P}}(|\log(h_{\tn{max}})|^{1/2}h_{\tn{max}}^{d/2+r+1}).
\end{equation}
Here, we have investigated  convergence in probability w.r.t. the sup-norm. However,  standard arguments show that this implies
the convergence which is investigated in Theorem \ref{a9}.

In a second step we find that the normalization with $c_j:=(\sqrt{g(t^j)}V_j)^{-1}$, $j=1,\hdots,2p$,
has no influence on the convergence as
translation and multiplication preserve the interval structure. More precisely, for any set 
$A=[a_1,b_1]\times\hdots\times[a_{2p},b_{2p}]\in\ca{A}'$ 
we have
\begin{align}\begin{split}\label{g2}
 &\big\{\big(c_jh_j^{d/2+r+1}B(F_j\sqrt{g})\big)_{j=1}^{2p}\in A\big\}
\\ =&\Big\{\big( h_j^{d/2+r+1}B(F_j\sqrt{g})\big)_{j=1}^{2p}\in[c_1^{-1}a_1, c_1^{-1}b_1]\times\hdots\times[c_{2p}^{-1}a_{2p}, c_{2p}^{-1}b_{2p}]\Big\},
\end{split}\end{align}
where $[c_1^{-1}a_1,c_1^{-1}b_1]\times\hdots\times[c_{2p}^{-1}a_{2p},c_{2p}^{-1}b_{2p}]$
still defines an element of the set $\ca{A}'$. A similar result holds for the normalization of the
test statistic. 

In a third step we show in Section \ref{proof58} that the normalization with the density estimator yields to a distribution-free limit process. We
 firstly assume that the density $g$ is known and prove
 \beq{u2}
  \max_{1\leq j\leq 2p}\Big| h_j^{d/2+r+1}\frac{B(F_j\sqrt{g})}{\sqrt{g(t^j)}V_j}-h_j^{d/2+r+1}\frac{B(F_j)}{V_j}\Big|
  =O_{\mathbb{P}}\big(\sqrt{h_{\tn{max}}\log(n)\log\log(n)}\big)=o_{\bb{P}}(1).
 \eeq
Hence,  by the consideration of the symmetric set $\ca{T}_n'$     it follows from (\ref{a21}), (\ref{h2}) and \eqref{u2} that
\begin{equation}\label{h3}
 \sup_{A\in\ca{A}}\Big|\bb{P}\Big(\Big( \frac{1}{\sqrt{n{g}(t_j)}V_j}|\sum_{i=1}^n
X_{i,j}|\Big)_{j=1}^p\in A\Big)-\bb{P}\Big(\Big(h_{j}^{d/2+r+1}\frac{|B(F_{j})|}{V_j}\Big)_{j=1}^p\in A\Big)\Big|=
o(1),
\end{equation}
 as for any real valued
random variable $X$ and any $a\in\bb{R}$ it holds
\[
\{|X|\in(-\infty,a]\}=\{X\in(-\infty,a]\}\cap\{-X\in(-\infty,a]\}.
\]

Next we insert the bandwidth normalization terms. To this end, we introduce the notation
\[
w(h)=\frac{\sqrt{\log(eh_{}^{-d}})}{\log\log(e^eh_{}^{-d})},\quad \td{w}(h)=\sqrt{2\log(h^{-d})}
\]
and write $w_j=w(h_j),\td{w}_j=\td{w}(h_j)$. Similar arguments as in \eqref{g2} show that the insertion of the bandwidth
correction terms has no influence on the convergence. Thus
recalling the definition of $\tilde X_j = w_j \big(h_{j}^{d/2+r+1}\frac{|B(F_{j})|}{V_j} - \tilde w_j\big) $ in \eqref{xtilde}
we obtain from
\eqref{h3}
\begin{equation}\label{h4}
 \sup_{A\in\ca{A}}\Big|\bb{P}\Big(\Big(
 w_j\Big(\frac{1}{\sqrt{ng(t_j)}V_j}
|\sum_{i=1}^nX_{i,j}|-\td w_j \Big)\Big)_{j=1}^p \in A  \Big)
-\bb{P}\Big( \tilde X \in A\Big)\Big|=
o(1),
\end{equation}

and  it remains to replace the true density by its estimator. For this purpose we show that
$$
 \max_{1\leq j\leq p}\Big| w_j\Big(\frac{1}{\sqrt{ng(t_j)}V_j}
|\sum_{i=1}^nX_{i,j}|-\td w_j\Big)-\td{X}_{j}^{(1)}\Big|=O_{\mathbb{P}}\Big(\frac{1}{\log\log(n)}\Big),
$$
where $\td{X}_j^{(1)}$ is defined in  \eqref {labelXj1}. Note that
\[
w_j\frac{1}{\sqrt{n}V_j}
|\sum_{i=1}^nX_{i,j}|\Big|\frac{1}{\sqrt{g(t^j)}}-\frac{1}{\sqrt{\hat{g}_n(t^j)}}\Big|\lesssim
w_j\frac{1}{\sqrt{ng(t^j)}V_j}|\sum_{i=1}^nX_{i,j}|\|g-\hat{g}_n\|_\infty
\]
almost surely by the boundedness from below of $g$ (and therefore of $\hat{g}_n$ almost surely). A null addition of the term $\td{w}_j$
shows that the latter is equal to
\[
  w_j\Big(\frac{1}{\sqrt{ng(t_j)}V_j}
|\sum_{i=1}^nX_{i,j}|-\td{w}_j\Big)\|g-\hat{g}_n\|_\infty+w_j\td{w}_j\|g-\hat{g}_n\|_\infty.
\]
The claim follows now from the convergence of $\big( w_j\big(\frac{1}{\sqrt{ng(t_j)}V_j}
|\sum_{i=1}^nX_{i,j}|-\td{w}_j\big)\big)_{j=1}^p$ proven in \eqref{h4} and the a.s. boundedness of the maximum of the limiting process 
proven in Section \ref{u3} below. 
 Note that we used the fact that
\[
h\mapsto\frac{{\log(eh^{-d}})}{\log\log(e^eh^{-d})}
\]
is decreasing in a neighborhood of $0$ (cf. \cite{MR3113812}, Lemma B.11).\\

\subsubsection{Proof of \eqref{a21}}\label{chern} 
The proof of \eqref{a21} mainly relies on Proposition 2.1 in \cite{Chernozhukov2014}. The result is stated as follows.
\begin{satz}\label{1}
 Let $X_1,\hdots,X_n$ be independent random vectors in $\bb{R}^{2p}$ with  $\bb{E}(X_{i,j})=0$ and $\bb{E}(X_{i,j}^2)<\infty$
 for $i=1,\hdots,n,\;j=1,\hdots,2p$. Moreover, let $Y'_1,\hdots,Y'_n$ be independent random vectors in $\bb{R}^{2p}$
 with $Y'_i\sim N(0,\bb{E}(X_iX_i^\top)),\;i=1,\hdots,n$. Let $b,q>0$ be some constants and let $B_n\geq1$ be a sequence of constants,
 possibly growing to infinity as $n\rightarrow\infty$.
Assume that the following conditions are satisfied:
\begin{enumerate}
 \item $n^{-1}\sum_{i=1}^n\bb{E}(X_{i,j}^2)\geq b$ for all $1\leq j\leq 2p$;
 \item $n^{-1}\sum_{i=1}^n\bb{E}(|X_{i,j}|^{2+k})\leq B_n^k$ for all $1\leq j\leq 2p$ and $k=1,2$;
 \item $\bb{E}\big(\big(\max_{1\leq j\leq 2p}|X_{i,j}|/B_n\big)^q\big)\leq 2$ for all $i=1,\hdots,n$.
\end{enumerate}

Then,
\[
 \sup_{A\in\ca{A}'}\Big|\bb{P}\Big( \frac{1}{\sqrt{n}}\sum_{i=1}^nX_i\in A\Big)-\bb{P}\Big(\frac{1}{\sqrt{n}}\sum_{i=1}^nY_i'\in A\Big)\Big|
 \leq C(D_n^{(1)}+D_{n,q}^{(2)}),
\]
where the sequences $D_n^{(1)}$ and $D_{n,q}^{(2)}$ are given by
\[
 D_n^{(1)}=\Big(\frac{B_n^2\log^7(2pn)}{n}\Big)^{1/6},\quad D_{n,q}^{(2)}=\Big(\frac{B_n^2\log^3(2pn)}{n^{1-2/q}}\Big)^{1/3}
\]
and the constant $C$ depends only on $b$ and $q$.
\end{satz}

For an application of Theorem \ref{1} we have to verify the condition \textit{(i)} and to find an appropriate
sequence $B_n$ for conditions \textit{(ii)} and \textit{(iii)}. For a proof of condition \textit{(i)}  notice that
\[
 \bb{E}(X_{1,j}^2)=h_j^{d+2r+2}\bb{E}\big((F_j(Y_1))^2\big)-h_j^{d+2r+2}\big(\bb{E}(F_j(Y_1))\big)^2\gtrsim
 h_j^{d+2r+2}\big(\bb{E}\big((F_j(Y_1))^2\big)-1\big),
\]
 where we used (\ref{h5}) in the inequality. Moreover, as the density of $g$ is bounded from below (Assumption \ref{q2}) we have
\bal
h_j^{d+2r+2}\bb{E}\big((F_j(Y_1))^2\big)&=h_j^{d+2r+2}\ird (F_j(x))^2g(x)\tn{d}x\\
&\gtrsim h_j^{d+2r+2} \int_{[-\delta,1+\delta]^d}(F_j(x))^2\tn{d}x\\
&= h_j^{d+2r+2}\ird(F_j(x))^2\tn{d}x- h_j^{d+2r+2}\int_{([-\delta,1+\delta]^d)^C}(F_j(x))^2\tn{d}x.
\end{align*}
In Lemma \ref{a11} \textit{(i)} we have proven that
$ \|F_j\|_{L^2(\bb{R}^d)}^2 \gtrsim h_j^{-d-2r-2},
$
and using the representation \eqref{a14} we obtain
\[
  \int_{([-\delta,1+\delta]^d)^C}(F_j(x))^2\tn{d}x\lesssim
  h_j^{-2d-2} \int_{([-\delta,1+\delta]^d)^C}\Big| \int_{\bb{R}^d}e^{iy.\frac{x-t^j}{h_j}}\frac{\ca{F}(\partial_{s^j}\phi)(y)}{\ol{\ca{F}(f_\ve)}(y/h_j)} \tn{d}y\Big|^2\tn{d}x.
\]
Moreover,  $[-t^j_1-\delta,-t^j_1+1+\delta]\times\hdots\times[-t^j_d-\delta,-t^j_d+1+\delta]\supseteq[-\delta,\delta]^d$ and a substitution show
\[
 \int_{([-\delta,1+\delta]^d)^C}\Big| \int_{\bb{R}^d}e^{iy.\frac{x-t^j}{h_j}}\frac{\ca{F}(\partial_{s^j}\phi)(y)}{\ol{\ca{F}(f_\ve)}(y/h_j)} \tn{d}y\Big|^2\tn{d}x\leq
   \int_{([-\delta,\delta]^d)^C}\Big| \int_{\bb{R}^d}e^{iy.\frac{x}{h_j}}\frac{\ca{F}(\partial_{s^j}\phi)(y)}{\ol{\ca{F}(f_\ve)}(y/h_j)} \tn{d}y\Big|^2\tn{d}x.
\]

We now follow the line of arguments presented in the proof of Lemma  \ref{2} \textit{(ii)} for $m=2$ and
note that by conducting integration by parts we get an additional factor $h_j^{d+1}$. Hence,
\beq{p1}
 \int_{([-\delta,1+\delta]^d)^C}(F_j(x))^2\tn{d}x\lesssim h_j^{-d-2r-1}.
\eeq
This concludes the proof of condition
\textit{(i)} as  $\bb{E}(X_{1,j}^2)\gtrsim 1-h_j-h_j^{d+2r+2}$ and $h_j\leq h_{\tn{max}}\rightarrow0$ for $n\rightarrow\infty$.

For a proof of condition \textit{(ii)} note that by part \textit{(ii)} of  Lemma \ref{2}  it follows that
\[
h_j^{(2+k)(d/2+r+1)} \bb{E}(|F_j(Y_1)|^{2+k})\lesssim h_{j}^{-kd/2} \tn{ for } k=1,2,
\]
and therefore $B_n$ can be chosen proportional to $h_{\tn{min}}^{-d/2}$.

An application of  Lemma \ref{2} \textit{(i)} yields
\[
 |X_{i,j}|\lesssim h_j^{-d/2}
\]
and therefore condition \textit{(iii)} of Theorem \ref{1} holds for any $q>0$ for the  choice of $B_n
=ch_{\tn{min}}^{-d/2}$, provided that the constant is chosen sufficiently large.

Hence, Theorem \ref{1} proves (recall that $p\leq n^K$)
\[
 \sup_{A\in\ca{A}'}\Big|\bb{P}\Big( \frac{1}{\sqrt{n}}\sum_{i=1}^nX_i\in A\Big)-\bb{P}\Big(\frac{1}{\sqrt{n}}\sum_{i=1}^nY_i'\in A\Big)\Big|
 \lesssim \Big(\frac{h_{\tn{min}}^{-d}\log^7(n)}{n}\Big)^{1/6}+\Big(\frac{h_{\tn{min}}^{-d}\log^3(n)}{n^{1-2/q}}\Big)^{1/3}
\]
for any $q>0$, which proves (\ref{a21}).\\

\subsubsection{Proof of \eqref{u2}}\label{proof58}
Define
\begin{align}\begin{split}\label{\"o1}
 R_j:=h_j^{d/2+r+1}\ird F_j(x)\big(\sqrt{g(x)}-\sqrt{g(t^j)}\big)\tn{d}B_x,
\end{split}
\end{align}
then the assertion follows from the statement
$$\max_{1\leq j\leq 2p}|R_j|=O_{\mathbb{P}}\big(\sqrt{h_{\tn{max}}\log(n)\log\log(n)}\big).$$
 Here, we used the fact that the constants $V_1,\hdots,V_{2p}$ are bounded uniformly from below (cf. Lemma \ref{a11}). 
 For this purpose, we will make use of a Slepian-type result.
Note that for all $\delta>0$
\begin{align}\begin{split}\label{u1}
 \bb{E}\big(R_j^2\big)=&\;h_j^{d+2r+2}\int_{[-\delta,1+\delta]^d}  \big(F_j(x)\big(\sqrt{g(x)}-\sqrt{g(t^j)}\big)\big)^2 dx\\&
+h_j^{d+2r+2}\int_{([-\delta,1+\delta]^d)^C}  \big(F_j(x)\big(\sqrt{g(x)}-\sqrt{g(t^j)}\big)\big)^2 dx. \end{split}\end{align}
For the first integral on the right hand side of \eqref{u1} we use the Lipschitz continuity of $g$ (Assumption \ref{q2}) and find
\[
 h_j^{d+2r+2}\int_{[-\delta,1+\delta]^d}  \big(F_j(x)\big(\sqrt{g(x)}-\sqrt{g(t^j)}\big)\big)^2\tn{d}x\lesssim
h_j^{d+2r+2} \int_{[-\delta,1+\delta]^d} \Big(F_j(x)\|x-t^j\|\frac{1}{2\sqrt{\xi}}\Big)^2\tn{d}x
\]
for some $\xi$ satisfying $|\xi - g(t^j)|\leq |g(x) - g(t^j)|$. If $\delta>0$ is sufficiently small, then $g$ is bounded from below on $[-\delta, 1+ \delta]^d$ (see the remark following   Assumption \ref{q2}), and  Lemma \ref{a11} \textit{(ii)} shows that an upper bound
of this term (up to some constant) is given by
\[
  h_j^{d+2r+2}\ird (F_j(x))^2\|x-t^j\|^2\tn{d}x\lesssim h_{\tn{max}}^{2} .
\]
The second integral on the right hand side of \eqref{u1} is bounded by $h_{\tn{max}}$
 which follows from \eqref{p1} and the boundedness of $g$ (Assumption \ref{q2}). Summarizing, we obtain  $$\bb{E}(R_j^2)\lesssim h_{\tn{max}}.$$
 Moreover, we can show by similar calculations as presented above and an  application of Lemma \ref{a11} \textit{(iv)} that
\[
 | \bb{E}\big(R_jR_k\big)|=(h_jh_k)^{d/2+r+1} \Big | \ird  F_j(x)\big(\sqrt{g(x)}-\sqrt{g(t^j)}\big)F_k(x)\big(\sqrt{g(x)}-\sqrt{g(t^k)}\big)\tn{d}x \Big | \lesssim h_{\tn{max}}.
\]
 Introducing  the random variables
\[
 \td{R_j}:=h_j^{d/2+r+2}\ird F_j(x)\tn{d}B_x,
\]
we obtain from Lemma \ref{a11} \textit{(i)}  and  \textit{(iii)}
\[
  \bb{E}\big(\td{R}_j^2\big) \lesssim h_{\tn{max}}^2, \bb{E}\big(\td{R}_j\td{R}_k\big)\lesssim h_{\tn{max}}^2.
\]
Hence,
\[
 \max_{1\leq j,k\leq 2p}\Big|\bb{E}\big((R_j-R_k)^2\big)-\bb{E}\big((\td{R}_j-\td{R}_k)^2\big)\Big|\lesssim h_{\tn{max}},
\]
and Theorem 2.2.5 in \cite{adler2007random} yields
\[
 \bb{E}\Big(\max_{1\leq j\leq 2p}R_j\Big)=\bb{E}\Big(\max_{1\leq j\leq 2p}\td{R}_j\Big)+
O\big(\sqrt{h_{\tn{max}}\log(n)}\big).
\]
Note that by the symmetry of the set $\ca{T}_n'$ with
respect to the direction we have $\bb{E}(\max_{1\leq j\leq 2p}R_j)=\bb{E}(\max_{1\leq j\leq 2p}|R_j|)$ and
 $\bb{E}(\max_{1\leq j\leq 2p}\td R_j)=\bb{E}(\max_{1\leq j\leq 2p}|\td R_j|)$, and we can consider expectations of positive random
 variables here.

 For an upper bound of $\bb{E}(\max_{1\leq j\leq 2p}\td{R}_j)$ we use the a.s. asymptotic boundedness of
\[
\max_{1\leq j\leq 2p}\frac{\sqrt{\log(eh_{j}^{-d}})}{\log\log(e^eh_{j}^{-d})}\Big(h_{j}^{-1}\frac{\td{R}_j}
{V_j}-\sqrt{2\log(h_{j}^{-d})}\Big)
\]
shown in Section \ref{u3} below, which implies
\[
 \bb{E}\Big(\max_{1\leq j\leq 2p}\td{R}_j\Big)=O\Big(\sqrt{\log(n)}h_{\tn{max}}\Big)
\]
and therefore $\bb{E}(\max_{1\leq j\leq 2p}R_j)=O(\sqrt{h_{\tn{max}}\log(n)})$. This
proves \eqref{u2} by an application of Markov's inequality.

\subsection{Boundedness of the approximating statistic}\label{u3}
In order to  prove that the approximating statistic $\max_{1\leq j\leq p}\td{X}_j$ considered in Theorem \ref{a9}
is almost surely bounded uniformly with respect to $n \in \bb{N} $ we note that for all $p\in\bb{N} $
$$\max_{1\leq j\leq p}\td{X}_j
\leq B ,
$$ where the random   variable $B$ is defined by
\[
B:=\sup_{(s,t,h)\in S^{d-1}\times[0,1]^d\times(0,1]}\frac{\sqrt{\log(eh_{}^{-d}})}{\log\log(e^eh_{}^{-d})}\Big(h_{}^{d/2+r+1}
\frac{|\int_{\bb{R}^d}F_{s,t,h}(x)\tn{d}B_x|}
{V_{s,t,h}}-\sqrt{2\log(h^{-d})}\Big),
\]
where the constant $V_{s,t,h}=h^{d/2+r+1}\|F_{s,t,h}\|_{L^2(\bb{R}^d)}$. $B$ does not depend on $n$ and we show below
that $B$ is almost surely bounded.
We  will make use of the following result (Theorem 6.1 and Remark 1, \cite{MR1833961}).

\begin{satz}\label{k5}
 Let $X$ be a stochastic process on a pseudometric space $(\ca{T},\rho)$ with continuous sample paths. Suppose that the
 following three conditions are satisfied.
 \begin{enumerate}
  \item There is a function $\sigma:\;\ca{T}\rightarrow(0,1]$ and a constant $K\geq 1$ such that
  \[
   \bb{P}\big(X(a)>\sigma(a)\eta\big)\leq K\exp(-\eta^2/2)\quad\text{for all }\eta>0\text{ and }a\in\ca{T}.
  \]
Moreover,
\[
 \sigma(b)^2\leq\sigma(a)^2+\rho(a,b)^2\quad\text{for all }a,b\in\ca{T}.
\]
\item For some constants $L,M\geq 1$,
\[
 \bb{P}\big(|X(a)-X(b)|>\rho(a,b)\eta\big)\leq L\exp(-\eta^2/M)\quad\text{for all }\eta>0\text{ and }a,b\in\ca{T}.
\]
\item For some constants $A,B,V>0$,
\[
 N\big((\delta u)^{{1}/{2}},\{a\in\ca{T}:\;\sigma(a)^2\leq \delta\}\big)\leq Au^{-B}\delta^{-V}\quad\text{for all }u,\delta\in(0,1],
\]
where $N(\ve,\ca{T}')$ denotes the packing number of the set $\ca{T}'\subseteq \ca{T}$.
 \end{enumerate}
Then, the random variable
\[
 \sup_{a\in\ca{T}}\Big(\frac{|X(a)|/\sigma(a)-(2\log(1/\sigma(a)^2))^{1/2}}
 {(\log(e/\sigma(a)^2))^{-{1}/{2}}\log\log(e^e/\sigma(a)^2)}\Big)
\]
is finite almost surely.
\end{satz}
For the application of Theorem \ref{k5} we introduce the pseudometric space $(\ca{T},\rho)$, 
where $\ca{T}=S^{d-1}\times[{0},{1}]^d\times(0,1]$ and 
\[
 \rho((s^1,t^1,h_1),(s^2,t^2,h_2))=\big(\|s^1-s^2\|_1^2+\|t^1-t^2\|+|h_1^d-h_2^d|\big)^{1/2}
\]
for $(s^1,t^1,h_1),(s^2,t^2,h_2)\in\ca{T}$. Moreover,  for $(s,t,h)\in\ca{T}$ define $\sigma(s,t,h)=h^{d/2}$,
\[
 X(s,t,h)=\sigma(s,t,h)\frac{h^{{d}/{2}+r+1}}{V_{s,t,h}}\int_{\bb{R}^d}F_{s,t,h}(x)\tn{d}B_x=\frac{h^{d+r+1}}{V_{s,t,h}}\int_{\bb{R}^d}F_{s,t,h}(x)\tn{d}B_x.
\]
In the following, we prove that the process $X$ fulfills the conditions of Theorem \ref{k5}.
\\

 \textit{(i)}: We have by definition of $\sigma$ and $\rho$ that
 \[
 \sigma(b)^2\leq\sigma(a)^2+\rho(a,b)^2\quad\text{for all }a,b\in\ca{T}.
\]
Furthermore, it holds
\[
\bb{P}\big(X(s,t,h)>\sigma(h)\eta\big)\leq\exp(-\eta^2/2)
\]
as $X(s,t,h)/\sigma(h)$ corresponds in distribution to a normal distributed random variable with mean zero and variance
 one by definition of $V_{s,t,h}$.\\

\textit{(ii)}: By definition, $X(s^1,t^1,h_1)-X(s^2,t^2,h_2)$ corresponds in distribution to a normal distributed
random variable with mean zero and variance
\[
\Big\|\frac{h_1^{d+r+1}}{V_{s^1,t^1,h_1}}F_{s^1,t^1,h_1}-\frac{h_2^{d+r+1}}{V_{s^2,t^2,h_2}}F_{s^2,t^2,h_2}\Big\|_{L^2(\bb{R}^d)}^2.
\]
W.l.o.g. we assume in the following $h_1\leq h_2$ and note that condition
\textit{(ii)}  (with $L=2$) follows from the inequality
\begin{align} \begin{split}\label{h6}
\Big\|\tfrac{h_1^{d+r+1}}{V_{s^1,t^1,h_1}}F_{s^1,t^1,h_1}-\tfrac{h_2^{d+r+1}}{V_{s^2,t^2,h_2}}F_{s^2,t^2,h_2}\Big\|_{L^2(\bb{R}^d)}
&\lesssim \big\|{h_1^{d+r+1}}{}F_{s^1,t^1,h_1}-{h_2^{d+r+1}}{}F_{s^2,t^2,h_2}\big\|_{L^2(\bb{R}^d)}\\&\quad\quad+h_1^{d/2}
|V_{s^1,t^1,h_1}-V_{s^2,t^2,h_2}|\\ 
&\lesssim \rho((s^1,t^1,h_1),(s^2,t^2,h_2))\end{split}
\end{align}
for $(s^1,t^1,h_1),(s^2,t^2,h_2)\in S^{d-1}\times[{0},{1}]^d\times(0,1]$. In the first inequality we used the fact that $V_{s^1,t^1,h_1}$ is uniformly bounded from below and
$\big\|h_1^{d+r+1}F_{s^1,t^1,h_1}\big\|_{L^2(\bb{R}^d)}\lesssim h_1^{d/2}$ as shown in
Lemma \ref{a11} \textit{(i)}.

In a proof of the second inequality in \eqref{h6} we note that by application of the triangle inequality 
\bal
 h_1^{d/2}|V_{s^1,t^1,h_1}-V_{s^2,t^2,h_2}|&=h_1^{d/2}\big|\|h_1^{d/2+r+1}F_{s^1,t^1,h_1}\|_{L^2(\bb{R}^d)}-\|h_2^{d/2+r+1}F_{s^2,t^2,h_2}\|_{L^2(\bb{R}^d)}\big|\\
 &\leq h_1^{d/2}\|h_1^{d/2+r+1}F_{s^1,t^1,h_1}-h_2^{d/2+r+1}F_{s^2,t^2,h_2}\|_{L^2(\bb{R}^d)}\\
 &\leq h_1^{d+r+1}\|F_{s^1,t^1,h_1}-F_{s^2,t^2,h_2}\|_{L^2(\bb{R}^d)}+\|F_{s^2,t^2,h_2}\|_{L^2(\bb{R}^d)}\big|h_1^{d+r+1}-h_1^{d/2}h_2^{d/2+r+1}\big|.
\end{align*}
In Lemma \ref{a11} \textit{(i)} we have proven $\|F_{s^2,t^2,h_2}\|_{L^2(\bb{R}^d)}\lesssim h_2^{-{d}/{2}-r-1}$, which implies 
\begin{align}
 \label{h7}
 h_1^{d/2}|V_{s^1,t^1,h_1}-V_{s^2,t^2,h_2}|&
 \lesssim h_1^{d+r+1}\|F_{s^1,t^1,h_1}-F_{s^2,t^2,h_2}\|_{L^2(\bb{R}^d)}+\big|\tfrac{h_1^{d+r+1}}{h_2^{d/2+r+1}}-h_1^{d/2}\big|\\
 &\lesssim h_1^{d+r+1}\|F_{s^1,t^1,h_1}-F_{s^2,t^2,h_2}\|_{L^2(\bb{R}^d)}+|h_1^{d/2}-h_2^{d/2}|.\nonumber
\end{align}
Moreover, we find by another application of the  inequality  $\|F_{s^2,t^2,h_2}\|_{L^2(\bb{R}^d)}\lesssim h_2^{-{d}/{2}-r-1}$
\begin{align}
\|h_1^{d+r+1}F_{s^1,t^1,h_1}-h_2^{d+r+1}F_{s^2,t^2,h_2}\|_{L^2(\bb{R}^d)}&\leq h_1^{d+r+1}\|F_{s^1,t^1,h_1}-F_{s^2,t^2,h_2}\|_{L^2(\bb{R}^d)}
\nonumber\\&\quad  ~~~~~~~+\|F_{s^2,t^2,h_2}\|_{L^2(\bb{R}^d)} \label{h88}
|h_1^{d+r+1}-h_2^{d+r+1}|   \\
\lesssim&\;
 h_1^{d+r+1}\|F_{s^1,t^1,h_1}-F_{s^2,t^2,h_2}\|_{L^2(\bb{R}^d)}+\big|\tfrac{h_1^{d+r+1}}{h_2^{{d}/{2}+r+1}}-h_2^{{d}/{2}}\big|
 \nonumber\\
 \lesssim&\;
  h_1^{d+r+1}\|F_{s^1,t^1,h_1}-F_{s^2,t^2,h_2}\|_{L^2(\bb{R}^d)}+|h_1^{d/2}-h_2^{{d/2}}|\nonumber. 
\end{align}
Hence, observing \eqref{h7} and \eqref{h88} the inequality \eqref{h6} follows from
\begin{equation} \label{h9}
 h_1^{d+r+1}\|F_{s^1,t^1,h_1}-F_{s^2,t^2,h_2}\|_{L^2(\bb{R}^d)}+|h_1^{d/2}-h_2^{{d/2}}|\lesssim \rho((s^1,t^1,h_1),(s^2,t^2,h_2)).
\end{equation}
For a proof of this inequality   we use Plancherel's theorem which yields
\bal
\|F_{s^1,t^1,h_1}-F_{s^2,t^2,h_2}\|_{L^2(\bb{R}^d)}^2\lesssim \int_{\bb{R}^d}\big(1+\|y\|^2)^{r}\Big|\ca{F}\Big(h_1^{-d}\partial_{s^1}\phi\big(
\tfrac{.-t^1}{h_1}\big)-h_2^{-d}\partial_{s^2}\phi\big(\tfrac{.-t^2}{h_2}\big)\Big)(y)\Big|^2\tn{d}y.
\end{align*}
The integrand on the right hand side can be estimated as follows
\bal \Big|\ca{F}\Big(h_1^{-d}\partial_{s^1}\phi\big(
\tfrac{.-t^1}{h_1}\big)-h_2^{-d}\partial_{s^2}\phi\big(\tfrac{.-t^2}{h_2}\big)\Big)(y)\Big|^2
\lesssim&\; \Big|\ca{F}\Big(h_1^{-d}\partial_{s^1}\phi\big(
\tfrac{.-t^1}{h_1}\big)-h_1^{-d}\partial_{s^2}\phi\big(\tfrac{.-t^1}{h_1}\big)\Big)(y)\Big|^2\\&\;+\Big|\ca{F}\Big(
h_1^{-d}\partial_{s^2}\phi\big(
\tfrac{.-t^1}{h_1}\big)-h_2^{-d}\partial_{s^2}\phi\big(\tfrac{.-t^2}{h_2}\big)\Big)(y)\Big|^2,\end{align*}
and we obtain
\bal &\|F_{s^1,t^1,h_1}-F_{s^2,t^2,h_2}\|_{L^2(\bb{R}^d)}^2\\\lesssim& \int_{\bb{R}^d}\big(1+\|y\|^2\big)^r\Big|
\sum_{k=1}^d\Big\{s^1_{k}\ca{F}\Big(h_1^{-d}\partial_{e^k}\phi\big(\tfrac{.-t^1}{h_1}\big)\Big)(y)
-s_k^2\ca{F}\Big(h_1^{-d}\partial_{e^k}\phi\big(\tfrac{.-t_1}{h_1}\big)\Big)(y)\Big\}\Big|^2\tn{d}y\\
&+\int_{\bb{R}^d}\big(1+\|y\|^2\big)^r\Big|\ca{F}\Big(h_1^{-d}\partial_{s^2}\phi\big(\tfrac{.-t^1}{h_1}\big)-h_2^{-d}\partial_{s^2}\phi
\big(\tfrac{.-t^2}{h_2}\big)\Big)(y)\Big|^2\tn{d}y,\end{align*}
where $e^k$   denotes the $k$th unit vector of $\bb{R}^d$ $(k=1,\hdots,d)$. By a substitution it follows that
\[\Big|\ca{F}\Big(h_1^{-d}\partial_{e^k}\phi\big(\tfrac{.-t^1}{h_1}\big)\Big)(y)\Big|=h_1^{-1}\big|\ca{F}(\partial_{e^k}\phi)(h_1y)\big|,\]
which gives
\begin{align}\begin{split}\label{a1}&\|F_{s^1,t^1,h_1}-F_{s^2,t^2,h_2}\|_{L^2(\bb{R}^d)}^2\\\lesssim& h_1^{-d-2r-2}\|s^1-s^2
\|_1^2\int_{\bb{R}^d}\big(1+\|y\|^2\big)^{r}\big|\ca{F}(\partial_{e^k}\phi)(y)\big|^2\tn{d}y\\
&+\ird \big(1+\|y\|^2\big)^r\Big|\ca{F}\Big(h_1^{-d}\partial_{s^2}\phi\big(\tfrac{.-t^1}{h_1}\big)\Big)(y)-
\ca{F}\Big(h_1^{-d}\partial_{s^2}
\phi\big(\tfrac{.-t^2}{h_1}\big)\Big)(y)\Big|^2\tn{d}y\\
&+\ird \big(1+\|y\|^2\big)^r\Big|\ca{F}\Big(h_1^{-d}\partial_{s^2}\phi\big(\tfrac{.-t^2}{h_1}\big)-h_2^{-d}\partial_{s^2}\phi\big(\tfrac{.-t^2}{h^2}
\big)\Big)(y)\Big|^2\tn{d}y.\end{split}\end{align}
Here, we used another substitution and the triangle inequality. For an upper bound for the first term on the right hand side of \eqref{a1}, note that
by Assumption \ref{a8} $\int_{\bb{R}^d}(1+\|y\|^2)^{r}|\ca{F}(\partial_{e^k}\phi)(y)|^2\tn{d}y$ is finite. Furthermore, a substitution within the
Fourier transform shows that the  second term of the  right hand side of \eqref{a1} is not greater than
\[\ird \big(1+\|y\|^2\big)^r\big|e^{-iy.t^1}-e^{-iy.t^2}\big|^2\Big|\ca{F}\Big(h_1^{-d}\partial_{s^2}
\phi\big(\tfrac{.}{h_1}\big)\Big)(y)\Big|^2\tn{d}y.\]
By an application of Euler's formula, $\cos(x)\geq 1-x$ for all $x\geq 0$ and  Cauchy-Schwartz's inequality, we find
\[
 \big|e^{-iy.t^1}-e^{-iy.t^2}\big|^2=\big|1-e^{-iy.(t^1-t^2)}\big|^2\lesssim\big(1+\|y\|^2)^{1/2}\|t^1-t^2\|.
\]
Therefore, two substitutions and Assumption \ref{a8} show that the second term on the right hand side of \eqref{a1} is bounded from
above  (up to some constant) by
\[\|t^1-t^2\|\ird \big(1+\|y\|^2\big)^{r+{1/2}}\Big|\ca{F}\Big(h_1^{-d}\partial_{s^2}
\phi\big(\tfrac{.}{h_1}\big)\Big)(y)\Big|^2\tn{d}y\lesssim h_1^{-d-2r-3}\|t^1-t^2\|.\]
It remains to consider the third term on the right hand side of \eqref{a1}. Plancherel's theorem, the rule for the Fourier transform
of a derivative and a substitution show that the third term on the right hand side of \eqref{a1} can be bounded by
\begin{align}\label{a2}
 &\sum_{|\bs{\alpha}|\leq \lceil r+1\rceil}\Big\|\partial^{\bs{\alpha}}\Big(h_1^{-d}\phi\big(\tfrac{.}{h_1}\big)-h_2^{-d}\phi\big(\tfrac{.}{h_2}\big)\Big)\Big\|^2_{L^2(\bb{R}^d)}\\
\nonumber\lesssim & \sum_{|\bs{\alpha}|\leq \lceil r+1\rceil}\Big\{\tfrac{1}{h_1^{2d+2|\bs{\alpha}|}}\big\|(\partial^{\bs{\alpha}}\phi)\big(\tfrac{.}{h_1}\big)
-(\partial^{\bs{\alpha}}\phi)\big(\tfrac{.}{h_2}\big)\big\|^2_{L^2(\bb{R}^d)}+\big\|(\partial^{\bs{\alpha}}
\phi)\big(\tfrac{.}{h_2}\big)\big\|_{L^2(\bb{R}^d)}^2\big|\tfrac{1}{h_1^{2d+2|\bs{\alpha}|}}-\tfrac{1}{h_2^{2d+2|\bs{\alpha}|}}\big|\Big\},
\end{align}
where we have used Assumption \ref{a8}. From the estimate $\|(\partial^{\bs{\alpha}}
\phi)(\tfrac{.}{h_2})\|_{L^2(\bb{R}^d)}^2\lesssim h_2^d$ we obtain that the second term on the right hand
side of \eqref{a2} is bounded from above (up to some constant) by
\[ h_2^d \big|\tfrac{1}{h_1^{2d+2|\bs{\alpha}|}}-\tfrac{1}{h_2^{2d+2|\bs{\alpha}|}}\big|\lesssim h_1^{-2d-2r-2}\big|h_1^d-h_2^d\big|
\]
for all $|\bs{\alpha}|\leq \lceil r+1\rceil$.  The first  term on the right hand
side of \eqref{a2} can be bounded by   Lemma \ref{a3} using Assumption \ref{a8}, that is
\[\tfrac{1}{h_1^{2d+2|\bs{\alpha}|}}\big\|(\partial^{\bs{\alpha}}\phi)
\big(\tfrac{.}{h_1}\big)-(\partial^{\bs{\alpha}}\phi)\big(\tfrac{.}{h_2}\big)\big\|^2_{L^2(\bb{R}^d)}\lesssim h_1^{-2d-2r-2}\big|h_1^d-h_2^d\big|
\]
for all $|\bs{\alpha}|\leq \lceil r+1\rceil $, which proves that the right hand side of \eqref{a2} is not greater (up to some constant) than
$h_1^{-2d-2r-2}\big|h_1^d-h_2^d\big|$.

Hence,
\[
 \|F_{s^1,t^1,h_1}-F_{s^2,t^2,h_2}\|_{L^2(\bb{R}^d)}^2\lesssim h_1^{-d-2r-2}\|s^1-s^2\|^2_1+h_1^{-d-2r-3}\|t^1-t^2\|+h_1^{-2d-2r-2}\big|h_1^d-h_2^d\big|
\]
 proves \eqref{h9} and concludes the proof of \textit{(ii)}.\\

\textit{(iii)}: Let $\td{N}(\ve,\ca{T}')\equiv\td{N}(\ve,\ca{T}',\rho)$ denote the covering number of the set $\ca{T}'\subseteq
\ca{T}$ and note that covering  and packing numbers are equivalent in the sense that
\[
 {N}(2\ve,\ca{T}')\leq\td{N}(\ve,\ca{T}')\leq{N}(\ve,\ca{T}').
\]
Hence, it suffices to find an upper bound for the cardinality of a well-chosen covering subset
$\ca{T}'\subset S^{d-1}\times[{0},{1}]^d\times\{h\in(0,1]:\;h^d\leq\delta \}$ that fulfills the following condition:

For any
$(s^1,t^1,h_1)\in S^{d-1}\times[{0},{1}]^d\times\{h\in(0,1]:\;h^d\leq\delta \}$ there exists $(s^2,t^2,h_2)\in\ca{T}'$
with $\rho^2((s^1,t^1,h_1),(s^2,t^2,h_2))\leq\delta u$. It is easy to see that such a set is given by
\beq{a4}
 \ca{T}'=\ca{T}_1'\times\ca{T}_2'\times\ca{T}_3',
\eeq
where $\ca{T}_1'$ is a covering subset of  $S^{d-1}$ with respect to $\sqrt\ve=\frac{(\delta u)^{1/2}}{\sqrt{3}}$ and
$\ca{T}_2',$ $\ca{T}_3'$ are  covering subsets of $[{0},$ ${1}]^d,$ $\{h\in(0,1]:\;h^d\leq\delta \}$, respectively,
with respect to $\ve=\frac{\delta u}{3}$. Here, the metrics under consideration are
$(s^2,s^1)\mapsto\|s^2-s^1\|_1,$ $(t^2,t^1)\mapsto\|t^2-t^1\|$ and $(h_2,h_1)\mapsto|h_2^d-h_1^d|$.

Again, we make use of the equivalence of packing and covering numbers and determine in the following upper bounds for the
packing numbers of  $S^{d-1}$ and $[{0}, {1}]^d$.

 We begin with the determination of an upper bound for the packing number  $N({\sqrt{\ve}},S^{d-1})$
 w.r.t. $\|\,.\,\|_1$ for $\ve>0$. Note
that by the equivalence of all norms in $\bb{R}^d$, the packing number  $N(\sqrt{\ve},S^{d-1})$ w.r.t. $\|\,.\,\|$ is of the
same order in $\ve$. We will therefore consider the latter.

Let $\ca{T}_1'$ be any subset of $S^{d-1}$ such that $\|s^2-s^1\|>{\sqrt{\ve}}$ for all $s^2,s^1\in \ca{T}_1',
\;s^2\neq s^1$. By definition of $\ca{T}_1'$, the open balls $B_{\frac{\sqrt{\ve}}{2}}(s^2)$ and $B_{\frac{\sqrt{\ve}}{2}}(s^1)$ are disjoint for all
$s^2,s^1\in \ca{T}_1',\;s^2\neq s^1$. Furthermore,
every  ball $B_{\frac{\sqrt{\ve}}{2}}(s),\;s\in\ca{T}_1'$, is contained in the annulus around the zero point with  radii $1+\frac{\sqrt{\ve}}{2}$
and $1-\frac{\sqrt{\ve}}{2}$. Recall that the volume of this annulus is of the order $(1+\frac{\sqrt{\ve}}{2})^d-(1-\frac{\sqrt{\ve}}{2})^d$.

A simple volume argument gives
\[
 \#\ca{T}_1'\lesssim \sqrt{\ve}^{-d}\Big(\big(1+\tfrac{\sqrt{\ve}}{2}\big)^d-\big(1-\tfrac{\sqrt{\ve}}{2}\big)^d\Big)\lesssim \ve^{({-d+1})/{2}}.
\]
It is a well-known fact that the packing number of $[{0},{1}]^d$ w.r.t. $\|\,.\,\|$ fulfills $N(\ve,[{0},{1}]^d)\lesssim\ve^{-d}$. Hence,
it remains to consider the covering number
$\td{N}(\ve,(0,\delta^{1/d}] )$ w.r.t.  the metric $(h_2,h_1)\mapsto{|h_2^d-h_1^d|}$. Observe that the distance  between adjacent points in the set
$\ca{T}_3':=\big\{(j\ve)^{1/d},\;j=1,\hdots,\lfloor\frac{\delta }{\ve}\rfloor\big\}$ is equal to $\ve$.
As a consequence, $\td{N}(\ve,(0,\delta^{1/d}] )\lesssim\frac{\delta }{\ve}$.

From \eqref{a4}  and the results presented above we deduce
\[
 N\big((\delta u)^\frac{1}{2},\{a\in\ca{T}:\;\sigma(a)^2\leq \delta\}\big)\lesssim u^{\frac{-3d-1}{2}}\delta^{\frac{-3d+1}{2}}.
\]

It remains to prove the continuity of the sample paths of $X$. For this purpose, we will make use of Theorem 1.3.5 in \cite{adler2007random}.

Define a further semimetric $\td d$  on $ {\ca{T}}$ by
\[
\td  d((s^1,t^1,h_1),(s^2,t^2,h_2))=\big(\bb{E}((X(s^1,t^1,h_1)-X(s^2,t^2,h_2))^2)\big)^{1/2}
\]
and the log-entropy $H(\ve)=\log(\td{N}(\ve, {\ca{T}},\td d))$. Then, Theorem 1.3.5 in \cite{adler2007random} states that $X$ has a.s. continuous sample
paths with respect to the semimetric $\td d$ if
\[
 \int_0^{\tn{diam}( {\ca{T}})/2}H^{1/2}(\ve)\tn{d}\ve<\infty,
\]
where $\tn{diam}( {\ca{T}})=\sup_{(s^1,t^1,h_1),(s^2,t^2,h_2)\in {\ca{T}}}\td d((s^1,t^1,h_1),(s^2,t^2,h_2))$. However, by the definition of $X$, we have that
\bal
\td  d((s^1,t^1,h_1),(s^2,t^2,h_2)) 
&=\|V_{s^1,t^1,h_1}^{-1}h_1^{d+r+1}F_{s^1,t^1,h_1}-V_{s^2,t^2,h_2}^{-1}h_2^{d+r+1}F_{s^2,t^2,h_2}\|_{L^2(\bb{R}^d)}\\
&\lesssim \rho((s^1,t^1,h_1),(s^2,t^2,h_2)),
\end{align*}
where the latter inequality has been proven in \textit{(ii)}. Hence, similar arguments as presented in \textit{(iii)} show that
$\td{N}(\ve, {\ca{T}},\td d)\lesssim\ve^{-a}$ for some $a>0$, which concludes the proof of the a.s. continuity of the sample paths of $X$
w.r.t. $\td d$ and implies the a.s. continuity of the sample paths of $X$ w.r.t. $\rho$.\\

\section{ Proofs of Theorems \ref{t10} and \ref{t14}}
\def\theequation{6.\arabic{equation}}
\setcounter{equation}{0}
\label{techres}

\begin{proof}[\textbf{Proof of Theorem \ref{t10}}]
Denote by $q$ the probability of at least one false rejection among
all tests  \eqref{t8} and \eqref{t81}. Using Theorem \ref{a9}, we further deduce from \eqref{t11}
\begin{align*}
q &=1-\bb{P}\Big( n^{-1}|\sum_{i=1}^nF_j(Y_i)|\leq\kappa^j_n(\alpha)\tn{ for all }j=1,\hdots,p\Big)\\
&=1-\bb{P}\big( \td{X}_j^{(1)}\leq \kappa_n(\alpha)\tn{ for all }j=1,\hdots,p\big)\\
&=1-\bb{P}\big(\td{X}_j\leq \kappa_n(\alpha)\tn{ for all }j=1,\hdots,p\big)+o(1)
 ~\leq~\alpha+o(1)\end{align*}
 for $n\rightarrow\infty.$
\end{proof}

\begin{proof}[\textbf{Proof of Theorem \ref{t14}}]
We begin deriving a criterion for the simultaneous rejection of the hypotheses \eqref{t5neu1} on a given set of scales. To this end, let $0<(\alpha_n)_{n\in\bb{N}}< 1$ be an arbitrary null sequence and
$J\subseteq\{1,\hdots,p\}$ be the set of all indices where the inequality
\beq{t20}
 \bb{E}(F_j(Y_1))=-\int_{\bb{R}^d} \partial_{s^j}f(x)\phi_{t^j,h_j}(x)\tn{d}x>2\kappa^j_n(\alpha_n)
\eeq
is satisfied. An application of Theorem \ref{a9} shows that the probability of simultaneous rejection of the Null Hypotheses for all tests in \eqref{t81}
indexed by $J$
 (where $\alpha$ is replaced by $\alpha_n$) is asymptotically equal to one, i.e.
 \[
  \td{q}:=\bb{P}\Big(n^{-1}\sum_{i=1}^nF_j(Y_i)>\kappa^j_n(\alpha_n)\tn{ for all }j\in J\Big)\geq 1-\alpha_n+o(1)=1-o(1).
 \]
Indeed,
\bal
\td{q}&\geq 
\bb{P}\Big(n^{-1}\sum_{i=1}^nF_j(Y_i)-\bb{E}(F_j(Y_1))\geq-\kappa^j_n(\alpha_n)\tn{ for all }j\in J\Big)\\
&\geq\bb{P}\Big(\Big |n^{-1}\sum_{i=1}^nF_j(Y_i)-\bb{E}(F_j(Y_1))\Big |\leq\kappa^j_n(\alpha_n)\tn{ for all }j\in J\Big)\\
&\geq 1-\alpha_n+o(1)
\end{align*}
by similar arguments as presented in the proof of Theorem \ref{t10}.

Now let $x^0\in(0,1)^d$ be a mode of $f$ and $(s,t,h)\in\ca{T}_n^{x^0}$, i.e.
$ch\geq\|x^0-t\|\geq2\sqrt{d} h$ for some $c>2\sqrt{d}$ and $\tn{angle}(x^0-t,s)\rightarrow0$ for $n\rightarrow\infty$. Following
the line of arguments presented in the proof of Theorem 3.3 in \cite{Eckle}, one can prove that, under the given assumptions,
$\partial_{s}f(x)\lesssim
-h$
for all $x\in \tn{supp}\phi_{t,h}$. Hence,
\[
 -\int_{\mathbb{R}^d} \phi_{t,h}(x)\partial_{s}f(x)\tn{d}x\gtrsim h.
\]
As $\kappa_n(\alpha_n)$ is uniformly bounded by Theorem \ref{a9}, we find that
\[
 \frac{h^{-d/2-r-1}}{\sqrt{n}}\Big(\frac{\log\log(e^eh^{-d})}{\sqrt{\log(eh^{-d}})}\kappa_n(\alpha_n)
+\sqrt{2\log(h^{-d})}\Big)\lesssim\frac{h^{-d/2-r-1}}{\sqrt{n}}\sqrt{\log(h^{-d})}.
\]
For a proof of \eqref{t20} it remains to find a condition on $h$ such that
\[
 h^{d/2+r+2}\gtrsim\frac{1}{\sqrt{n}}\sqrt{\log(h^{-d})},
\]
which holds for $h\geq C\log(n)^{1/(d+2r+4)}n^{-1/(d+2r+4)}$ for some $C>0$ sufficiently large.
\end{proof}

\section{Two technical results}\label{sec7}
\def\theequation{7.\arabic{equation}}
\setcounter{equation}{0}
\begin{lemma}\label{a3}
 Let $\Phi:\bb{R}^d\rightarrow\bb{R}$ be continuously differentiable with compact support. Then,
 \[
  \big\|\Phi\big(\tfrac{.}{h_1}\big)-\Phi\big(\tfrac{.}{h_2}\big)\big\|^2_{L^2(\bb{R}^d)}\lesssim\big|h_1^d-h_2^d\big|
 \]
 for all $h_1,h_2\in(0,1].$
\end{lemma}
\begin{proof}[\textbf{Proof of Lemma \ref{a3}}]
 W.l.o.g. we assume in the following that $h_1\leq h_2$ and obtain
 \begin{align}\begin{split}\label{9.11}
& \ird \Big(\Phi\big(\tfrac{x}{h_1}\big)-\Phi\big(\tfrac{x}{h_2}\big)\Big)^2\tn{d}x\\=&\ird \Phi^2\big(\tfrac{x}{h_1}\big)\tn{d}x
 +\ird \Phi^2\big(\tfrac{x}{h_2}\big)\tn{d}x -2\ird \Phi\big(\tfrac{x}{h_1}\big)\Phi\big(\tfrac{x}{h_2}\big)\tn{d}x\\
 =&h_1^d\ird \Phi^2(x)\tn{d}x+h_2^d\ird \Phi^2(x)\tn{d}x-2h_1^d\ird \Phi(x)\Phi\big(\tfrac{h_1}{h_2}x\big)\tn{d}x.
 \end{split}\end{align}
 Observe that
 \[
  \Phi\big(\tfrac{h_1}{h_2}x\big)=\Phi(x)+\big(-1+\tfrac{h_1}{h_2}\big)x.\nabla \Phi(\xi)
 \]
for some $\xi$ on the line that connects $x$ and $\tfrac{h_1}{h_2} x$. Hence, the  term in \eqref{9.11} is bounded  by
\bal
 &\big(h_2^d-h_1^d\big)\ird \Phi^2(x)\tn{d}x+2h_1^d\big|1-\tfrac{h_1}{h_2}\big|\sup_{y\in\tn{supp}\Phi}\|\nabla \Phi(y)\|\ird |\Phi(x)|\|x\|\tn{d}x\\
\lesssim& \big(h_2^d-h_1^d\big)+h_1^d-\frac{h_1^{d+1}}{h_2}\lesssim h_2^d-h_1^d.
\end{align*}

\end{proof}
\begin{lemma}[Fa\`{a} di Brunos formula]\label{s5}
Let $k\in\bb{N}$ and assume that $h_1,h_2:\bb{R}\rightarrow\bb{R}$ are sufficiently smooth functions. Then,
\begin{equation}
\frac{d^k}{dx^k} h_1(h_2(x)) = \sum_{(m_1,...,m_k) \in \mathcal{M}_k}
\frac{k!}{m_1!...m_k!} h_1^{(m_1+...+m_k)}(h_2(x)) \prod_{j=1}^k \left( \frac{h_2^{(j)}(x)}{j!} \right)^{m_j}
\end{equation}
for every $x\in\bb{R}$, where $\mathcal{M}_k$ is the set of all $k$-tuples of non-negative integers satisfying $\sum_{j=1}^k j m_j = k$.

\end{lemma}

\end{document}